\documentclass[11pt]{article}

\usepackage[T1]{fontenc}    % use 8-bit T1 fonts
\usepackage{hyperref}       % hyperlinks
\usepackage{url}            % simple URL typesetting
\usepackage{booktabs}       % professional-quality tables
\usepackage{amsfonts}       % blackboard math symbols
\usepackage{nicefrac}       % compact symbols for 1/2, etc.
\usepackage{microtype}      % microtypography

\usepackage[compact]{titlesec}
\usepackage{enumitem}
\usepackage{wrapfig}

\usepackage{amsmath,amstext,amsopn,amsthm}
\usepackage{mathtools}
\usepackage{cleveref}
\usepackage{graphicx}
\usepackage{color}

\usepackage{bbm}
\usepackage{algorithm}

\usepackage[noend]{algpseudocode}
\newtheorem{theorem}{Theorem}
\newtheorem{corollary}{Corollary}
\newtheorem{lemma}{Lemma}

\newtheorem{definition}{Definition}

\newtheorem{proposition}{Proposition}

\newtheorem{observation}{Observation}

\makeatletter
\def\BState{\State\hskip-\ALG@thistlm}
\makeatother

\def\RR{{\mathbb R}}

\def\ZZ{{\mathbb Z}}

\def\K{{\mathcal K}}

\DeclareMathOperator{\argmax}{argmax}

\DeclareMathOperator{\OPT}{OPT}

\DeclareMathOperator{\ident}{\bf I}
\DeclareMathOperator{\SigmaB}{\bf \Sigma}
\DeclareMathOperator{\muB}{\bf \mu}

\usepackage{mathtools}
\renewcommand{\algorithmicrequire}{\textbf{Input:}}
\renewcommand{\algorithmicensure}{\textbf{Output:}}

\usepackage{sidecap}

\allowdisplaybreaks

\usepackage[margin=1in]{geometry}

\usepackage{thmtools}
\usepackage{thm-restate}

\usepackage{hyperref} % hyperlinks
\usepackage[round]{natbib}
\title{On the Unreasonable Effectiveness of the Greedy Algorithm: \\ Greedy Adapts to Sharpness}
\usepackage{times}

\author{Alfredo Torrico\thanks{Polytechnique Montreal, Montreal. Email: torrico.alf@gmail.com} \quad Mohit Singh\thanks{Georgia Institute of Technology, Atlanta. Email: mohitsinghr@gmail.com} \quad Sebastian Pokutta\thanks{Zuse Institute Berlin (ZIB), TU Berlin. Email: pokutta@zib.de}}

\date{}

\begin{document}

\maketitle

\begin{abstract}
Submodular maximization has been widely studied over the past decades, mostly because of its numerous applications in real-world problems. It is well known that the standard greedy algorithm guarantees a worst-case approximation factor of $1-1/e$ when maximizing a monotone submodular function under a cardinality constraint. However, empirical studies show that its performance is substantially better in practice. This raises a natural question of explaining this improved performance of the greedy algorithm. 

In this work, we define \emph{sharpness} for submodular functions as a candidate explanation for this phenomenon. The sharpness criterion is inspired by the concept of strong convexity in convex optimization. We show that the greedy algorithm provably performs better as the sharpness of the submodular function increases. This improvement ties in closely with the faster convergence rates of first order methods for sharp functions in convex optimization. Finally, we perform a computational study to empirically support our theoretical results and show that sharpness explains the greedy performance better than other justifications in the literature.
\end{abstract}

%\begin{keywords}

%\end{keywords}

\section{Introduction}\label{sec:introduction}
During the last decade, the interest in constrained submodular maximization has increased significantly, especially due to its numerous applications in real-world problems. Common examples of these applications are influence modeling in social networks \citep{kempe_etal03}, sensor placement \citep{krause2008robust}, document summarization \citep{lin_etal09}, or in general constrained feature selection \citep{krause_etal05,das2008algorithms,krause_etal08a,krause_etal09,powers_etal16}.  To illustrate the submodular property, consider a simple example of selecting the most influential nodes $S$ in a social network where information is seeded at $S$ and is passed around in the network based on a certain stochastic process. Submodularity captures the natural property that the total number of nodes influenced marginally decreases as more nodes are seeded \citep{kempe_etal03}. Given the importance of submodular optimization, there has been significant progress in  designing new algorithms with provable guarantees \citep{calinescu2011maximizing,ene_etal16,buchbinder_etal16,sviridenko_04}. %Some lines of research have also focused on improving the running time of the standard algorithms used in these problems, see e.g. \cite{minoux_78,badanidiyuru_vondrak14,mirzasoleiman_etal15}.

The most fundamental problem in submodular optimization is to maximize a monotonically increasing submodular function subject to a cardinality constraint. A classical result~\citep{nemhauser1978analysis,nemhauser1978hardness} shows that the greedy algorithm is a multiplicative $(1-1/e)$-approximation algorithm. Moreover, no other efficient algorithm can obtain a better guarantee \citep{nemhauser1978hardness}. However, empirical observations have shown that standard algorithms such as the greedy algorithm performs considerably better in practice. Explaining this phenomenon has been a tantalizing challenge. Are there specific properties in real world instances that the greedy algorithm exploits? An attempt to explain this phenomenon has been made with the concept of \emph{curvature}~\citep{conforti_etal84}. In simple words, this parameter measures how close to linear the objective function is. This line of work establishes a (best possible) approximation ratio of $1-\gamma/e$ using curvature $\gamma \in [0,1]$ as parameter \citep{sviridenko_etal15}.

In this work, we focus on giving an explanation for those instances in which the optimal solution clearly stands out over the rest of feasible solutions. For this, we consider the concept of \emph{sharpness} initially introduced in continuous optimization \citep{lojas63} and we adapt it to submodular optimization. Roughly speaking, this property measures the behavior of the objective function around the set of optimal solutions. Sharpness in continuous optimization translates in faster convergence rates. Equivalently, we will show that the greedy algorithm for submodular maximization performs better as the sharpness of the objective function increases, as a discrete analog of ascent algorithms in continuous optimization. 

Our main contributions in this work are: (1) to introduce the \emph{sharpness} criteria in submodular optimization as a novel candidate explanation of the performance of the greedy algorithm; (2) to show that the standard greedy algorithm automatically adapt to the sharpness of the objective function, without requiring this information as part of the input; (3) to provide provable guarantees that depend only on the sharpness parameters; and (4) to empirically support our theoretical results with a detailed computational study in real-world applications.

%modeling migration \citep{procaccia19}

\subsection{Problem Formulation}
In this work, we study the \emph{submodular function maximization problem subject to a single cardinality constraint}. %We extend our results to a general class of structured combinatorial constraints such as \emph{partition constraints} and \emph{gammoids}, formally called \emph{matroids}.
Formally, consider a ground set of $n$ elements $V=\{1,\ldots, n\}$ and a non-negative set function $f:2^V\to\RR_+$. We denote the marginal value for any subset $A\subseteq V$ and $e\in V$ by $f_A(e):=f(A+e)-f(A)$, where $A+e:= A\cup\{e\}$. A set function $f$ is {\it submodular} if, and only if, it satisfies the {\it diminishing returns property}. Namely, for any $e\in V$ and $A\subseteq B\subseteq V\backslash\{e\}$, $f_A(e)\geq f_B(e)$.
We say that $f$ is \emph{monotone} if for any $A\subseteq B\subseteq V$, we have $f(A)\leq f(B)$. To ease the notation, we will write the value of singletons as $f(e) := f(\{e\})$.  For simplicity, we assume throughout this work that $f$ is \emph{normalized}, i.e., $f(\emptyset)=0$. Our results still hold when $f(\emptyset)\neq0$, but an additive extra term has to be carried over.

As we mentioned before, our work is mostly focused on the optimization of non-negative monotone submodular functions subject to a single cardinality constraint. In this setting, we are given a non-negative integer $k$ and the goal is to optimally select a subset $S$ that contains at most $k$ elements of $V$. Formally, the optimization problem is the following
\begin{equation}\label{eq:problem_def}
\max\{f(S): \ |S|\leq k\}. \tag{P$_1$}
\end{equation}
Throughout the rest of the paper, we denote the optimal value as $\OPT$. In this context, we assume the \emph{value oracle model}, i.e., the decision-maker queries the value of $S$ and the oracle returns $f(S)$. It is well known that \eqref{eq:problem_def} is NP-hard to solve exactly under the value oracle model. Therefore, most of the literature has focused on designing algorithms with provable guarantees. A natural approach is the standard \emph{greedy algorithm} which constructs a set by adding in each iteration the element with the best marginal value while maintaining feasibility \citep{fisher1978analysis}. The authors show that the greedy algorithm achieves a $1-1/e$ approximation factor for problem \eqref{eq:problem_def}, which is tight \citep{nemhauser1978hardness,feige1998threshold}. We give a detailed description of the related work in Section \ref{sec:related_work}. Even though the best possible guarantee is $1-1/e$, the standard greedy algorithm usually performs better in practice.

To explain this phenomenon we adapt the concept of \emph{sharpness} used in continuous optimization \citep{hoffman52,lojas63,polyak79,lojas93,bolte_etal07}. The notion of sharpness is also known as H\"{o}lderian error bound on the distance to the set of optimal solutions. Broadly speaking, this property characterizes the behavior of a function around the set of optimal solutions. To exemplify this property, consider a concave function $F$, a feasible region $X$, a set of optimal solutions $X^*=\argmax_{x\in X}F(x)$ and a distance function $d(\cdot, X^*):X \to\RR_+$. Then, $F$ is said to be $(c,\theta)$-sharp if for any $x\in X$
\[
F^* - F(x) \geq \left(\frac{d(x,X^*)}{c} \right)^{1/\theta},
\]
where $F^*=\max_{x\in X}F(x)$. Since we assumed $F$ is concave, we know that $\nabla F(x)\cdot(x^*-x)\geq F^* - F(x)$, therefore we have
\[
\nabla F(x)\cdot(x^*-x) \geq \left(\frac{d(x,X^*)}{c} \right)^{1/\theta}.
\]
Sharpness has been widely used to study convergence rates in convex and non-convex optimization, see e.g., \citep{nemirovski_nesterov85,karimi_etal16,bolte_etal14,roulet_daspremont17,kerdreux_etal18}. For a detailed review on the sharpness condition in continuous optimization, we refer the interested reader to \citep{roulet_daspremont17}.

\subsection{Our Contributions and Results}

Our main contribution is to introduce multiple concepts of sharpness in submodular optimization that mimic conditions previously studied in continuous optimization. We show that the greedy algorithm performs better than the worst-case guarantee $1-1/e$ for functions that are sharp with appropriate parameters. Empirically, we obtain improved guarantees on real data sets when using more refined conditions.

\subsubsection{Monotonic Sharpness}
Given parameters $c\geq 1$ and $\theta\in[0,1]$, we define monotonic sharpness as follows.
\begin{definition}[Monotonic Sharpness]\label{def:d_sharp}
A non-negative monotone submodular set function $f:2^V\to\RR_+$ is said to be \emph{$(c,\theta)$-monotonic sharp}, if there exists an optimal solution $S^*$ for Problem \eqref{eq:problem_def} such that for any subset $S\subseteq V$ with $|S|\leq k$ the function satisfies
\begin{equation}\label{eq:sharp_inequality}
\sum_{e\in S^*\backslash S} f_S(e)\geq \left(\frac{|S^*\backslash S|}{k\cdot c}\right)^{\frac{1}{\theta}}\cdot \OPT.
%f(S\cup S^*) - f(S)\geq \left[\frac{|S^*\backslash S|}{k\cdot c}\right]^{\frac{1}{\theta}}\cdot f(S^*)
\end{equation}
\end{definition}
The property can be interpreted as implying that the optimal set $S^*$ is not just unique but any solution which differs significantly from $S^*$ has substantially lower value. Our first main result for Problem \eqref{eq:problem_def} is stated in the next theorem.

\begin{theorem}\label{theorem:main_card}
Consider a non-negative monotone submodular function $f:2^V\to\RR_+$ which is $(c,\theta)$-monotonic sharp. Then, the greedy algorithm returns a feasible set $S^g$ for \eqref{eq:problem_def} such that
\[f(S^g)\geq \left[1-\left(1-\frac{\theta}{c}\right)^{\frac{1}{\theta}}\right] \cdot f(S^*).\]
\end{theorem}

%Given parameters $c\geq 1$ and $\theta\in[0,1]$, we define submodular sharpness as follows.
%\begin{definition}[Submodular Sharpness]\label{def:d_sharp}
%A non-negative monotone submodular function $f:2^V\to\RR_+$ is said to be \emph{$(c,\theta)$-sharp}, if there exists an optimal solution $S^*$ for Problem \eqref{eq:problem_def} such that for any subset $S\subseteq V$ with $|S|\leq k$ the function satisfies
%\begin{equation}\label{eq:sharp_inequality}
%\sum_{e\in S^*\backslash S} f_S(e)\geq \left(\frac{|S^*\backslash S|}{k\cdot c}\right)^{\frac{1}{\theta}}\cdot \OPT.
%%f(S\cup S^*) - f(S)\geq \left[\frac{|S^*\backslash S|}{k\cdot c}\right]^{\frac{1}{\theta}}\cdot f(S^*)
%\end{equation}
%\end{definition}
%The property can be interpreted as implying that the optimal set $S^*$ is not just unique but any solution which differs significantly from $S^*$ has substantially lower value. Our first main result for Problem \eqref{eq:problem_def} is stated in the next theorem.
%
%
%\begin{theorem}\label{theorem:main_card}
%Consider a non-negative monotone submodular function $f:2^V\to\RR_+$ which is $(c,\theta)$-sharp. Then, the greedy algorithm returns a feasible set $S^g$ for \eqref{eq:problem_def} such that
%\[f(S^g)\geq \left[1-\left(1-\frac{\theta}{c}\right)^{\frac{1}{\theta}}\right] \cdot f(S^*).\]
%\end{theorem}

We remark that \emph{any} monotone set function $f$ is $(c,\theta)$-monotonic sharp as $c\gets 1$ and $\theta \gets 0$ (Lemma \ref{lemma:sharpness_facts} in Section \ref{sec:discrete_sharpness}). Corollary \ref{corollary:recovery} in Section \ref{sec:discrete_sharpness} shows that the guarantee $1-\left(1-\theta/c\right)^{1/\theta}$ tends to $1-1/e$ when $(c,\theta)$ goes to $(1,0)$, recovering the classical guarantee for any monotone submodular function~\citep{nemhauser1978analysis}. However, if the parameters $(c,\theta)$ are bounded away from $(1,0)$, we obtain a strict improvement over this worst case guarantee. In Section~\ref{sec:experiments}, we show experimental results to illustrate that real data sets do show improved parameters. In Section \ref{sec:discrete_sharpness}, we also discuss the sharpness of simple submodular functions such as linear and concave over linear functions.

Definition \ref{def:d_sharp} can be considered as a \emph{static} notion of sharpness, since parameters $c$ and $\theta$ do not change with respect to the size of $S$. We generalize this definition by considering the notion of \emph{dynamic monotonic sharpness}, in which the parameters $c$ and $\theta$ depend on the size of the feasible sets, i.e., $c_{|S|}\geq 1$ and $\theta_{|S|}\in[0,1]$. This allows us to obtain improved guarantees for the greedy algorithm based on how the monotonic sharpness changes dynamically. Formally, we define dynamic sharpness as follows.

\begin{definition}[Dynamic Monotonic Sharpness]\label{def:dynamic_sharp}
%A non-negative monotone submodular function $f:2^V\to\RR_+$ is said to be \emph{dynamic $(c,\theta)$-sharp}, if for any subset $S\subseteq V$ with $|S|\leq k$ the function satisfies
%\begin{align*}
%\sum_{e\in S^*\backslash S} f_S(e)\geq \left[\frac{|S^*\backslash S|}{k\cdot c_{|S|}}\right]^{\frac{1}{\theta_{|S|}}}\cdot f(S^*)
%\end{align*}
A non-negative monotone submodular function $f:2^V\to\RR_+$ is said to be \emph{dynamic $(c,\theta)$-monotonic sharp}, where $c=(c_0,c_1,\ldots,c_{k-1})\in[1,\infty)^k$ and $\theta=(\theta_0,\theta_1,\ldots,\theta_{k-1})\in[0,1]^k$, if there exists an optimal solution $S^*$ for Problem \eqref{eq:problem_def} such that for any subset $S\subseteq V$ with $|S|\leq k$ the function satisfies
\begin{align*}
\sum_{e\in S^*\backslash S} f_S(e)\geq \left(\frac{|S^*\backslash S|}{k\cdot c_{|S|}}\right)^{\frac{1}{\theta_{|S|}}}\cdot f(S^*).
\end{align*}
\end{definition}

In other words, we say that a function $f$ is $(c_i,\theta_i)$-monotonic sharp for any subset such that $|S| = i$, where $i\in\{0,\ldots,k-1\}$. Note that since we have $k$ pairs of parameters $(c_i,\theta_i)$, then there are at most $k-1$ intervals in which sharpness may change. If the parameters are identical in every interval, then we recover Definition \ref{def:d_sharp} of monotonic sharpness. We obtain the following guarantee for dynamic monotonic sharpness.

\begin{theorem}\label{theorem:general_approx}
 Consider a non-negative monotone submodular function $f:2^V\to\RR_+$ that is dynamic $(c,\theta)$-sharp with parameters $c=(c_0,c_1,\ldots,c_{k-1})\in[1,\infty)^k$ and $\theta=(\theta_0,\theta_1,\ldots,\theta_{k-1})\in[0,1]^k$. Then, the greedy algorithm returns a set $S^g$ for \eqref{eq:problem_def} such that
\begin{equation*}
f(S^g)\geq \Bigg[1 - \Bigg(\bigg(\Big(1-\frac{\theta_0}{c_0k}\Big)^{\frac{\theta_1}{\theta_0}} - \frac{\theta_1}{c_1k}\bigg)^{\frac{\theta_2}{\theta_1}} -  \cdots - \frac{\theta_{k-1}}{c_{k-1}k}\Bigg)^{\frac{1}{\theta_{k-1}}}\Bigg] \cdot f(S^*).
\end{equation*}
\end{theorem}

In Section~\ref{sec:experiments}, we empirically show that the guarantees shown in Theorem \ref{theorem:general_approx} strictly outperforms the factors provided by Theorem \ref{theorem:main_card}. 

\subsubsection{Submodular Sharpness}

Definition \ref{def:d_sharp} measures the distance between a feasible set $S$ and $S^*$ as the cardinality of its difference. However, this distance may not be a precise measure, since the value $f(S)$ could be quite close to $\OPT$. Therefore, we introduce the concept of \emph{submodular sharpness} as a natural generalization of Definition \ref{def:d_sharp}. Given parameters $c\geq 1$ and $\theta\in[0,1]$, we define submodular sharpness as follows
\begin{definition}[Submodular Sharpness]\label{def:total_d_sharp}
A non-negative monotone submodular function $f:2^V\to\RR_+$ is said to be \emph{$(c,\theta)$-submodular sharp}, if there exists an optimal solution $S^*$ for Problem \eqref{eq:problem_def} such that for any subset $S\subseteq V$ with $|S|\leq k$ the function satisfies
\begin{equation}\label{eq:total_sharp_inequality}
\max_{e \in S^*\backslash S} f_S(e) \geq \frac{1}{kc}\left[f(S^*) - f(S) \right]^{1-\theta}\OPT^\theta
%f(S\cup S^*) - f(S)\geq \left[\frac{|S^*\backslash S|}{k\cdot c}\right]^{\frac{1}{\theta}}\cdot f(S^*)
\end{equation}
\end{definition}

%Recall that \emph{any} arbitrary non-negative monotone submodular function has curvature $\gamma =1$. Therefore, when the curvature is unknown, we assume that $\gamma =1$ is implicit and simply call the function $(c,\theta)$-submodular sharp. 
This inequality can be interpreted as the submodular version of the Polyak-\L{}ojasiewicz inequality \citep{polyak63,lojas63} with the $\ell_\infty$-norm. We observe that any non-negative monotone submodular function that is $(c,\theta)$-monotonic sharp is also $(c,\theta)$-submodular sharp (see Lemma \ref{lemma:subm_sharpness_facts} in Section \ref{sec:s_sharpness}). We obtain the following main result for submodular sharpness.

\begin{theorem}\label{theorem:total_sharp}
Consider a non-negative monotone submodular function $f:2^V\to\RR_+$ which is $(c,\theta)$-submodular sharp. Then, the greedy algorithm returns a feasible set $S^g$ for \eqref{eq:problem_def} such that
\[f(S^g)\geq \left[1-\left(1-\frac{\theta}{c}\right)^{\frac{1}{\theta}}\right] \cdot f(S^*).\]
%Moreover, if the function has curvature $\gamma\in[0,1]$, then the greedy algorithm returns a feasible set $S^g$ for \eqref{eq:problem_def} such that
%\[f(S^g)\geq \frac{1}{\gamma}\left[1-\left(1-\frac{\theta\cdot\gamma}{c}\right)^{\frac{1}{\theta}}\right] \cdot f(S^*).\]
\end{theorem}

Observe that \emph{any} monotone submodular function $f$ is $(c,\theta)$-submodular sharp as $c\gets 1$ and $\theta \gets 0$ (see Lemma \ref{lemma:subm_sharpness_facts} in Section \ref{sec:s_sharpness}). Thus, Theorem \ref{theorem:total_sharp} recovers the classical guarantee $1-1/e$~\citep{nemhauser1978analysis}. Since Definition \ref{def:total_d_sharp} is weaker than Definition \ref{def:d_sharp}, the approximation factor guaranteed by Theorem \ref{theorem:total_sharp} is at least as good as in Theorem \ref{theorem:main_card}. More importantly, we will empirically show in Section \ref{sec:experiments} that there is a significant improvement in the approximation guarantees when using Definition \ref{def:total_d_sharp}.%, and even more when the curvature is considered.

Finally, similar to dynamic monotonic sharpness, we introduce the concept of \emph{dynamic submodular sharpness}. 

\begin{definition}[Dynamic Submodular Sharpness]\label{def:total_dynamic_sharp}
A non-negative monotone submodular function $f:2^V\to\RR_+$ is said to be \emph{dynamic $(c,\theta)$-submodular sharp}, where $c=(c_0,c_1,\ldots,c_{k-1})\in[1,\infty)^k$ and $\theta=(\theta_0,\theta_1,\ldots,\theta_{k-1})\in[0,1]^k$, if there exists an optimal solution $S^*$ for Problem \eqref{eq:problem_def} such that for any subset $S\subseteq V$ with $|S|\leq k$ the function satisfies
\begin{equation}\label{eq:dynamic_subm_sharp_inequality}
\max_{e \in S^*\backslash S} f_S(e) \geq \frac{1}{kc_{|S|}}\left[f(S^*) - f(S)\right]^{1-\theta_{|S|}}f(S^*)^{\theta_{|S|}}
\end{equation}
\end{definition}

Finally, we obtain the following result for Problem \eqref{eq:problem_def}.

\begin{theorem}\label{theorem:general_approx2}
 Consider a non-negative monotone submodular function $f:2^V\to\RR_+$ that is dynamic $(c,\theta)$-submodular sharp with parameters $c=(c_0,c_1,\ldots,c_{k-1})\in[1,\infty)^k$ and $\theta=(\theta_0,\theta_1,\ldots,\theta_{k-1})\in[0,1]^k$. Then, the greedy algorithm returns a set $S^g$ for \eqref{eq:problem_def} such that
\begin{equation*}
f(S^g)\geq \Bigg[1 - \Bigg(\bigg(\Big(1-\frac{\theta_0}{c_0k}\Big)^{\frac{\theta_1}{\theta_0}} - \frac{\theta_1}{c_1k}\bigg)^{\frac{\theta_2}{\theta_1}} -  \cdots - \frac{\theta_{k-1}}{c_{k-1}k}\Bigg)^{\frac{1}{\theta_{k-1}}}\Bigg] \cdot f(S^*).
\end{equation*}
%Moreover, if the function has curvature $\gamma\in[0,1]$, then the greedy algorithm returns a feasible set $S^g$ for \eqref{eq:problem_def} such that
%\begin{equation*}
%f(S^g)\geq \frac{1}{\gamma}\cdot\Bigg[1 - \Bigg(\bigg(\Big(1-\frac{\theta_0\cdot\gamma}{c_0k}\Big)^{\frac{\theta_1}{\theta_0}} - \frac{\theta_1\cdot\gamma}{c_1k}\bigg)^{\frac{\theta_2}{\theta_1}} -  \cdots - \frac{\theta_{k-1}\cdot\gamma}{c_{k-1}k}\Bigg)^{\frac{1}{\theta_{k-1}}}\Bigg] \cdot f(S^*).
%\end{equation*}
\end{theorem}

We prove Theorems \ref{theorem:total_sharp} and \ref{theorem:general_approx2} in Section \ref{sec:s_sharpness}. %Although, their proofs are similar to Theorems \ref{theorem:main_card} and \ref{theorem:general_approx}.

\paragraph{Experimental Results.} In Section \ref{sec:experiments}, we provide a computational study in real-world applications such as movie recommendation, non-parametric learning, and clustering. Our objective is to experimentally contrast our theoretical results with the existing literature, such as the concepts of \emph{curvature} \citep{conforti_etal84} and \emph{submodular stability} \citep{chatzia_etal17}. While all these results try to explain the improved performance of greedy, sharpness provides an alternate explanation for this improved behavior.

\subsection{Related Work}\label{sec:related_work}

As remarked earlier, the greedy algorithm gives a $(1-\frac1e)$-approximation for maximizing a submodular function subject to a cardinality constraint~\citep{nemhauser1978analysis} and is optimal~\citep{feige1998threshold,nemhauser1978hardness}.

 The concept of \emph{curvature}, introduced in \citep{conforti_etal84}, measures how far the function is from being linear. Formally, a monotone submodular function $f:2^V\to\RR_+$ has \emph{total curvature} $\gamma \in [0,1]$ if
\begin{equation}\label{eq:curvature}
\gamma = 1 - \min_{e\in V^*}\frac{f_{V-e}(e)}{f(e)},
\end{equation}
where $V^* = \{e \in V: \ f(e)>0\}$. For a submodular function with total curvature $\gamma\in[0,1]$, \citet{conforti_etal84} showed that the greedy algorithm guarantees an approximation factor of $(1-e^{-\gamma})/\gamma$.  %In the same work, the author defines the concept of \emph{curvature with respect to a set} as follows: given a fixed set $S\subseteq V$, a submodular function $f:2^V\to\RR_+$ has curvature $\gamma_S$ with respect to $S$, if $\gamma_S\in[0,1]$ is the smallest value such that for any $T\subseteq V$
%\begin{equation}\label{eq:curvature_vondrak}
%f(S\cup T) - f(S) + \sum_{e\in S\cap T}f_{S\cup T - e}(e) \geq (1 - \gamma_S)f(T).
%\end{equation}
%Note that this notion is weaker than the total curvature defined by \citet{conforti_etal84}, since a function with total curvature $\gamma$ has curvature $\gamma_S = \gamma$ with respect to any set $S\subseteq V$. The author finally shows that when a monotone submodular function has curvature $\gamma$ with respect to the optimal solution then $(1-e^{-\gamma})/\gamma$ is the best possible for a single matroid constraint. This sequence of results finishes with \citep{sviridenko_etal15}. The authors provide a \emph{modified continuous greedy} algorithm to obtain approximate solutions for the problem of maximizing the sum of a monotone submodular function and a linear function. This algorithm consists of two steps: guessing the value of the linear function in the optimal solution and then run the continuous greedy algorithm in a smaller feasible region determined by this guess. \citet{sviridenko_etal15} show that for a monotone submodular function with total curvature $\gamma\in[0,1]$, this algorithm achieves an approximation factor of $1-\gamma/e$ for a single matroid. Moreover, in the same work they show that this result is the best possible under the definition given by \citet{conforti_etal84}. One drawback of the \emph{modified continuous greedy} is the guessing step since it may affect the running time, but recently a guess-free algorithm which ensures the same guarantee was shown in \citep{feldman18} . 
The notion of curvature has been also used when minimizing submodular functions \citep{iyer_etal13}, and the equivalent notion of \emph{steepness} in supermodular function minimization \citep{ilev01}; as well as maximizing submodular functions under general combinatorial constraints~\citep{conforti_etal84,sviridenko_04,sviridenko_etal15,feldman18}; we refer to the interested reader to the literature therein for more details.

\paragraph{Stability.} Close to our setting is the concept of stability widely studied in discrete optimization. Broadly speaking, there are instances in which the unique optimal solution still remains unique even if the objective function is slightly perturbed. For example, the concept of \emph{clusterability} has been widely studied in order to show the existence of \emph{easy instances} in clustering \citep{balcan_etal09,daniely_etal12}, we refer the interested reader to the survey \citep{ben_david15a} for more details. Stability has been also studied in other contexts such as influence maximization \citep{he_etal14}, Nash equilibria \citep{balcan_etal17}, and Max-Cut \citep{bilu_linial12}. Building on \citep{bilu_linial12}, the concept of stability under multiplicative perturbations in submodular optimization is studied in \citep{chatzia_etal17}. Formally, given a non-negative monotone submodular function $f$,  $\tilde{f}$ is a $\gamma$-perturbation if: (1) $\tilde{f}$ is non-negative monotone submodular; (2) $f\leq \tilde{f}\leq \gamma\cdot f$; and (3) for any $S\subseteq V$ and $e\in V\backslash S$, $0\leq \tilde{f}_S(e) - f_S(e)\leq (\gamma-1)\cdot f(e)$. Now, assume we have an instance of problem \eqref{eq:problem_def} with a unique optimal solution, then this instance is said to be $\gamma$-stable if for any $\gamma$-perturbation of the objective function, the original optimal solution remains being unique. \citet{chatzia_etal17} show that the greedy algorithm recovers the unique optimal solution for 2-stable instances. However, it is not hard to show that 2-stability for problem  \eqref{eq:problem_def} is a strong assumption, since 2-stable instances can be easily solved by maximizing the sum of singleton values and thus $2$-stable functions do not capture higher order relationship among elements. %They also show recovery results for more general constraints called independence systems.
%
%\paragraph{Sharpness in continuous optimization.} The concept of \emph{sharpness} has been mainly studied in continuous optimization. It is also known as H\"{o}lderian error bound on the distance to the set of optimal solutions \citep{hoffman52,lojas63,polyak79,lojas93,bolte_etal07}. Broadly speaking, this property characterizes the behavior of a function around the set of optimal solutions. Formally, if $X^*\subseteq X$ is the set of optimal solutions in a universe of feasible points $X$ and $d(\cdot, X^*):X \to\RR_+$ is a distance function, then a function $f$ is said to be $(c,\theta)$-sharp if for any $x\in X$
%\begin{equation}\label{eq:convex_sharpness}
%d(x,X^*) \leq c(f^* - f(x))^\theta,
%\end{equation}
%where $f^*$ is the optimal objective value of a maximization problem. Sharpness has been widely used to study convergence rates in convex and non-convex optimization, see e.g. \citep{nemirovski_nesterov85,karimi_etal16,bolte_etal14,roulet_daspremont17,kerdreux_etal18}. For a detailed review on the sharpness condition in continuous optimization, we refer the interested reader to \citep{roulet_daspremont17}.

\section{Analysis of Monotonic Sharpness}\label{sec:discrete_sharpness}
In this section, we focus on the analysis of the standard greedy algorithm for Problem \eqref{eq:problem_def} when the objective function is $(c,\theta)$-monotone sharp. First, let us prove the following basic facts:
\begin{lemma}\label{lemma:sharpness_facts}
Consider any monotone set function $f:2^V\to\RR_+$. Then,
\begin{enumerate}
\item There is always a set of parameters $c$ and $\theta$ such that $f$ is $(c,\theta)$-monotonic sharp. In particular, $f$ is always $(c,\theta)$-monotonic sharp when both $c\to1$ and $\theta\to0$.
\item If $f$ is $(c,\theta)$-monotonic sharp, then for any $c'\geq c$ and $\theta'\leq \theta$, $f$ is $(c',\theta')$-monotonic sharp. Therefore, in order to maximize the guarantee of Theorem \ref{theorem:main_card} we look for the smallest feasible $c$ and the largest feasible $\theta$.
\item If $f$ is also submodular, then Inequality \eqref{eq:sharp_inequality} needs to be checked only for sets of size exactly $k$.
%\item Inequality \eqref{eq:sharp_inequality} needs to be checked only for sets of size exactly $k$.
\end{enumerate}
\end{lemma}
\begin{proof}
\begin{enumerate}
\item Note that $\frac{|S^*\backslash S|}{k} \leq 1$, so
\(\left(\frac{|S^*\backslash S|}{k\cdot c}\right)^{\frac{1}{\theta}}\leq \left(\frac{1}{c}\right)^{\frac{1}{\theta}}, \)
which shows that $\left(\frac{|S^*\backslash S|}{k\cdot c}\right)^{\frac{1}{\theta}}\to 0$ when $c\to 1$ and $\theta \to 0$. Therefore, Definition \ref{def:d_sharp} is simply
\(\sum_{e\in S^*\backslash S} f_S(e)\geq 0,\)
which is satisfied since from monotonicity we have $f_S(e)\geq 0$.
\item Observe that $\left(\frac{|S^*\backslash S|}{k\cdot c}\right)^{\frac{1}{\theta}}$ as a function of $c$ and $\theta$ is increasing in $\theta$ and decreasing in $c$. Therefore, $\left(\frac{|S^*\backslash S|}{k\cdot c}\right)^{\frac{1}{\theta}}\geq \left(\frac{|S^*\backslash S|}{k\cdot c'}\right)^{\frac{1}{\theta'}}$ for  $c'\geq c$ and $\theta'\leq \theta$.
\item Given that $f$ is submodular, then $f_A(e)\geq f_B(e)$ for any $e\in V$ and $A\subseteq B\subseteq V\backslash e$, so in \eqref{eq:sharp_inequality} we can assume without loss of generality that $|S| = k$.
\end{enumerate}
\end{proof}

We are ready to prove Theorem \ref{theorem:main_card}.  We emphasize that the greedy algorithm automatically adapts to the sharpness of the function and does not require explicit access to the sharpness parameters in order to obtain the desired guarantees. For completeness, let us recall the standard greedy algorithm.
\begin{algorithm}[h!]
\caption{Greedy \citep{nemhauser1978analysis}}\label{alg:greedy}
\begin{algorithmic}[1]
\renewcommand{\algorithmicrequire}{\textbf{Input:}}
\renewcommand{\algorithmicensure}{\textbf{Output:}}
\Require ground set $V=\{1,\ldots, n\}$, monotone submodular function $f:2^V\to\RR_+$, and $k\in\ZZ_+$.
\Ensure feasible set $S$ with $|S|\leq k$.
\State Initialize $S=0$.
\While {$|S| < k$}
\State $S \leftarrow S+ \argmax_{e\in V\backslash S}f_S(e)$
\EndWhile
\end{algorithmic}
\end{algorithm}

Recall that given parameters $c\geq 1$ and $\theta\in [0,1]$, a function is $(c,\theta)$-monotonic sharp if there exists an optimal set $S^*$ such that for any set $S$ with at most $k$ elements, then
\[\sum_{e\in S^*\backslash S} f_S(e)\geq \left(\frac{|S^*\backslash S|}{k\cdot c}\right)^{\frac{1}{\theta}}\cdot f(S^*)
\]
\begin{proof}[Proof of Theorem \ref{theorem:main_card}]
Let us denote by $S_{i}$ the set we obtain in the $i$-th iteration of Algorithm \ref{alg:greedy}. Note that $S^g:= S_k$. By using the properties of the function $f$, we can obtain the following sequence of inequalities
\begin{align}
f(S_i) - f(S_{i-1}) &= \frac{\sum_{e\in S^*\backslash S_{i-1}}f(S_i) - f(S_{i-1})}{|S^*\backslash S_{i-1}|} \label{eq:ineq_greedy}\\
&\geq \frac{\sum_{e\in S^*\backslash S_{i-1}}f_{S_{i-1}}(e)}{|S^*\backslash S_{i-1}|} \tag{choice of greedy}
%&\geq \frac{f(S^*)^\theta}{k} \cdot \frac{\sum_{e\in S^*\backslash S_{i-1}}f_{S_{i-1}}(e)}{\left(\sum_{e\in S^*\backslash S_{i-1}}f_{S_{i-1}}(e)\right)^\theta} \tag{$(\theta,c)$-sharp}\\
\end{align}
Now, from the sharpness condition we know that
\[\frac{1}{|S^*\backslash S_{i-1}|}\geq \frac{f(S^*)^\theta}{kc}\cdot \left(\sum_{e\in S^*\backslash S_{i-1}}f_{S_{i-1}}(e)\right)^{-\theta}\]
so we obtain the following bound
\begin{align}
f(S_i) &- f(S_{i-1})  \geq \frac{f(S^*)^\theta}{kc}\cdot \left(\sum_{e\in S^*\backslash S_{i-1}}f_{S_{i-1}}(e)\right)^{1-\theta} \tag{sharpness} \\ \vspace{0.5em}
&\geq \frac{f(S^*)^\theta}{kc}\cdot [f(S_{i-1}\cup S^*)-f(S_{i-1})]^{1-\theta} \tag{submodularity}\\
&\geq \frac{f(S^*)^\theta}{kc}\cdot [f(S^*) - f(S_{i-1})]^{1-\theta}. \tag{monotonicity}
\end{align}
Therefore, we need to solve the following recurrence
\begin{equation}\label{eq:recurrence_aux} a_i \geq a_{i-1} + \frac{a^\theta}{kc}\cdot [a - a_{i-1}]^{1-\theta}\end{equation}
where $a_i =f(S_i)$, $a_0=0$ and $a=f(S^*)$.

%In the Appendix, we prove that
%\[a_i\geq \left[1-\left(1-\frac{\theta }{c}\cdot\frac{i}{k}\right)^{\frac{1}{\theta}}\right]\cdot a, \]
%for any $i\in[k]$, from which the proof easily follows noting the bound for $a_k=f(S^g)$.

Define \(h(x) = x+ \frac{a^\theta}{kc}\cdot [a-x]^{1-\theta},\) where $x\in[0,a]$. Observe that
\(h'(x)=1-\frac{a^\theta(1-\theta)}{kc}\cdot[a-x]^{-\theta}.\)
Therefore, $h$ is increasing in the interval \(I:= \left\{x: \ 0\leq x\leq a\cdot \left(1-\left(\frac{1-\theta}{kc}\right)^{1/\theta}\right) \right\}\). Let us define
\[b_i := a\cdot \left[1-\left(1-\frac{\theta }{c}\cdot\frac{i}{k}\right)^{\frac{1}{\theta}}\right] .\]
First, let us check that $b_i\in I$ for all $i\in\{0,\ldots, k\}$. Namely, for any $i$ we need to show that
\begin{equation*}
a\cdot \left[1-\left(1-\frac{\theta }{c}\cdot\frac{i}{k}\right)^{\frac{1}{\theta}}\right]  \leq   a\cdot \left(1-\left(\frac{1-\theta}{kc}\right)^{1/\theta}\right) \quad \Leftrightarrow \quad (kc-i\theta)\geq 1-\theta
\end{equation*}
%Since $1-\delta \leq 1-\delta +\delta kc$, then it is sufficient to prove that
%\[1-(1-\delta)^{\frac{1}{\theta}}\left(1-\frac{\theta }{c}\cdot\frac{i}{k}\right)^{\frac{1}{\theta}}  \leq   1-\left(\frac{1-\theta}{kc}\right)^{1/\theta} \quad \Leftrightarrow \quad (1-\delta)(kc-i\theta)\geq 1-\theta\]
The expression $kc -i \theta$ is decreasing on $i$. Hence, we just need the inequality for $i=k$, namely $k(c-\theta)\geq 1-\theta$, which is true since $c\geq1$ and $k\geq 1$.

Our goal is to prove that $a_i\geq b_i$, so by induction let us assume that $a_{i-1}\geq b_{i-1}$ is true (observe $a_0\geq b_0$). By using monotonicity of $h$ on the interval $I$, we get $h(a_{i-1})\geq h(b_{i-1})$. Also, observe that recurrence \eqref{eq:recurrence_aux} is equivalent to write $a_i \geq h(a_{i-1})$ which implies that $a_i \geq h(b_{i-1})$. To finish the proof we will show that $h(b_{i-1})\geq b_i$.

Assume for simplicity that $a=1$. For $x\in [0,k]$, define
 \[g(x) :=  \left(1-\frac{\theta}{kc}\cdot x\right)^{1/\theta}.\]
 Note that $g'(x) = - \frac{1}{kc}\cdot g(x)^{1-\theta}$ and $g''(x) = \frac{1-\theta}{(kc)^2}\cdot g(x)^{1-2\theta}$. Observe that $g$ is convex, so for any $x_1,x_2\in[0,k]$ we have \(g(x_2)\geq g(x_1) + g'(x_1)\cdot (x_2-x_1).\) By considering $x_2 = i$ and $x_1 = i-1$, we obtain
\begin{equation}\label{eq:convex_ineq1}
g(i) - g(i-1) - g'(i-1) \geq 0
\end{equation}
On the other hand,
\begin{align*}
h(b_{i-1})- b_i&=1-\left(1-\frac{\theta }{c}\cdot\frac{i-1}{k}\right)^{\frac{1}{\theta}} + \frac{1}{kc}\cdot \left(1-\frac{\theta }{c}\cdot\frac{i-1}{k}\right)^{\frac{1-\theta}{\theta}} - 1+\left(1-\frac{\theta }{c}\cdot\frac{i}{k}\right)^{\frac{1}{\theta}}\\
&=\left(1-\frac{\theta }{c}\cdot\frac{i}{k}\right)^{\frac{1}{\theta}} -\left(1-\frac{\theta }{c}\cdot\frac{i-1}{k}\right)^{\frac{1}{\theta}} +  \frac{1}{kc}\cdot \left(1-\frac{\theta }{c}\cdot\frac{i-1}{k}\right)^{\frac{1-\theta}{\theta}}
\end{align*}
which is exactly the left-hand side of \eqref{eq:convex_ineq1}, proving $h(b_{i-1})- b_i\geq 0$.

Finally, \[f(S^g) = a_k\geq b_k =  \left[1-\left(1-\frac{\theta }{c}\right)^{\frac{1}{\theta}}\right]\cdot f(S^*),\] proving the desired guarantee.
\end{proof}

We recover the classical $1-1/e$ approximation factor, originally proved in \citep{nemhauser1978analysis}.
\begin{corollary}\label{corollary:recovery}
The greedy algorithm achieves a $1-\frac1e$-approximation for any monotone submodular function.
\end{corollary}

\begin{proof}
We proved in Lemma \ref{lemma:sharpness_facts} that any monotone submodular function is $(c,\theta)$-monotonic sharp when $\theta\to0$ and $c\to1$. On the other hand, we know that $1-\left(1-\frac{\theta }{c}\right)^{\frac{1}{\theta}}$ is increasing on $\theta$, so for all $c\geq 1$ and $\theta\in[0,1]$ we have
\(1-\left(1-\frac{\theta }{c}\right)^{\frac{1}{\theta}}\geq 1-e^{-1/c}\). By taking limits, we obtain $\lim_{c\to 1, \theta\to 0}1-\left(1-\frac{\theta }{c}\right)^{\frac{1}{\theta}}\geq 1-e^{-1}$.
\end{proof}

%\begin{observation}\label{obs:best_linear}
%Due to Proposition \ref{lemma:linear}, we can easily obtain the best pair $(c,\theta)$ for a linear function by solving
%\begin{align*}\max & \quad 1-\left(1-\frac{\theta}{c}\right)^{\frac{1}{\theta}} \\ \quad s.t. \ & \quad c\geq \frac{\ell}{k}\cdot\left(\frac{W(\ell)}{W(k)}\right)^{-\theta}, \quad \forall \ \ell \in [k-1] \\ &\quad c\geq 1, \ \theta \in[0,1]. \end{align*}
%A similar optimization problem can be stated for a concave over modular function by using Proposition \ref{lemma:concave}.
%\end{observation}

\paragraph{Contrasting Sharpness with Curvature.} A natural question is how our result compares to the curvature analysis proposed in \citep{conforti_etal84}. Specifically, is there any class of functions in which the monotonic sharpness criterion provides considerable better approximation guarantees than the curvature analysis? Recall the definition of curvature given in Section \ref{sec:related_work}. \citet{conforti_etal84} proved that if a monotone submodular function has curvature $\gamma\in[0,1]$, then the standard greedy algorithm guarantees a $(1-e^{-\gamma})/\gamma$ fraction of the optimal value. Consider an integer $k\geq 2$, a ground set $V=[n]$ with $n=2k$, a set function $f(S) = \min\{|S|,k+1\}$ and problem \eqref{eq:problem_def}. Observe that any set of size $k$ is an optimal set, so consider $S^*$ any set with $k$ elements. Note also that the curvature of the function is $\gamma = 1$, since $f(V) = k+1$ and $f(V- e) = k+1$ for any $e\in V$. Therefore, by using the curvature analysis, we can conclude an approximation factor of $1-1/e$. Let us analyze the monotonic sharpness of this function. Pick any subset $S\subseteq V$ such that $|S|\leq k$ and $S\neq S^*$, then we have $f(S) = |S|$ and $f_S(e) = 1$ for any $e\in S^*\backslash S$. Hence,
\[\frac{\sum_{e\in S^*\backslash S} f_S(e)}{f(S^*)} =  \frac{|S^*\backslash S|}{k},\]
which implies that parameters $\theta =1$ and $c=1$ are feasible in inequality \eqref{eq:sharp_inequality}. We can conclude that the sharpness analysis gives us an approximation factor of 1. From this simple example, we observe that curvature is a global parameter of the function that does not take into consideration the optimal set and can be easily perturbed, while the sharpness criterion focuses on the behavior of the function around the optimal solution. More precisely, take any function $f$ with curvature close to 0, which means an approximation guarantee close to 1. Then, take $\tilde{f}(S) = \min\{f(S),f(S^*)\}$. This function is still monotone and submodular, but its curvature now is 1, while its sharpness is the same as the original function $f$.

\subsection{Dynamic Monotonic Sharpness}\label{sec:dynamic_sharp}
In this section, we focus on proving the main results for dynamic sharpness, Theorem \ref{theorem:general_approx}. We emphasize that the greedy algorithm automatically adapts to the dynamic sharpness of the function without requiring parameters $(c_i,\theta_i)$ as part of the input.

\begin{proof}[Proof of Theorem \ref{theorem:general_approx}]
%The proof in this case is similar to the proof of Theorem \ref{theorem:main_card}, but it has to be separated in each step $i$ of the recursion. Let us recall recursion \eqref{eq:recurrence_aux}: for any $i\in[k]$
%\[a_i \geq a_{i-1} + \frac{a^\theta}{kc_{i-1}}\cdot [a - a_{i-1}]^{1-\theta_{i-1}},\]
%where $a_i = f(S_i)$, $a_0 = 0$, and $a = f(S^*)$. For simplicity assume that $a=1$. The proof consists in proving by induction the following inequality for all $i\in[k]$
%\begin{align*}
%a_i &\geq \Bigg[1 - \Bigg(\bigg(\Big(1-\frac{\theta_0}{c_0k}\Big)^{\frac{\theta_1}{\theta_0}} \\
%& \qquad \qquad \cdots -\frac{\theta_{i-2}}{c_{i-2}k}\bigg)^{\theta_{i-1}/\theta_{i-2}}- \frac{\theta_{i-1}}{c_{i-1}k}\Bigg)^{\frac{1}{\theta_{i-1}}}\Bigg]
%\end{align*}
Observe that in the $i$-th iteration of the greedy algorithm $|S_i| = i$, so the change of the sharpness parameters will occur in every iteration $i$. The proof is similar to Theorem \ref{theorem:main_card}, but the recursion needs to be separated in each step $i$. Let us recall recursion \eqref{eq:recurrence_aux}: for any $i\in[k]$
\[a_i \geq a_{i-1} + \frac{a^\theta}{kc_{i-1}}\cdot [a - a_{i-1}]^{1-\theta_{i-1}},\]
where $a_i = f(S_i)$, $a_0 = 0$, and $a = f(S^*)$. For simplicity assume that $a=1$. %The only difference is that now the sharpness parameters that may vary in each step of the greedy algorithm.
%We assume there exists a single switching moment $t_0\in\{0,\ldots, k-1\}$. Therefore, for any step $i\in\{0,\ldots, t_0-1\}$ of the greedy algorithm we have $c_i = c_0$ and $\theta_i = \theta_0$. This implies that for any step $i\in\{1,\ldots, t_0\}$ we have our first recurrence
%\begin{equation}\label{eq:first_rec}
%a_i \geq a_{i-1} + \frac{a^{\theta_0}}{kc_0}\cdot [a - a_{i-1}]^{1-\theta_0}
%\end{equation}
We proceed the proof by induction. Note that for $i=1$ we need to prove that
\(a_1\geq \frac{1}{kc_0}.\)
For $c_0$ and $\theta_0$, the sharpness inequality \eqref{eq:sharp_inequality} needs to be checked only for $S=\emptyset$, which is trivially satisfied with $c_0=\theta_0= 1$. From the proof of Theorem \ref{theorem:main_card} we can conclude the following for $i=1$.
\[a_{1}\geq  \left[1-\left(1-\frac{\theta_0 }{kc_0}\right)^{\frac{1}{\theta_0}}\right], \]
and given that $c_0 = \theta_0=1$ is valid pair of parameters, then this inequality is simply \(a_1\geq \frac{1}{kc_0}\), proving the desired base case. Let us denote
\begin{equation*}
b_j := \Bigg[1 - \Bigg(\bigg(\Big(1-\frac{\theta_0}{c_0k}\Big)^{\frac{\theta_1}{\theta_0}} \cdots -\frac{\theta_{j-2}}{c_{j-2}k}\bigg)^{\theta_{j-1}/\theta_{j-2}}- \frac{\theta_{j-1}}{c_{j-1}k}\Bigg)^{\frac{1}{\theta_{j-1}}}\Bigg]
\end{equation*}
for $1\leq j\leq k$. We assume that $a_{i}\geq b_{i}$ is true, and will prove that $a_{i+1}\geq b_{i+1}$.
% It remains to analyze the rest of the recurrence for steps $i\in\{t_0,\ldots,k-1\}$. During these steps we have $c_i = c_1$ and $\theta_i = \theta_1$. Therefore, we have the following recurrence for steps $i\in\{t_0+1,\ldots,k\}$
%\begin{equation}\label{eq:second_rec}
%a_i \geq a_{i-1} + \frac{a^{\theta_1}}{kc_1}\cdot [a - a_{i-1}]^{1-\theta_1}.
%\end{equation}
%We will show that the solution for \eqref{eq:second_rec} satisfies
%\begin{equation}\label{eq:sol_thm_dynamic}a_i\geq \left[1 - \left(\left(1-\frac{\theta_0t_0}{c_0k}\right)^{\frac{\theta_1}{\theta_0}} + \frac{\theta_1t_0}{c_1k}-\frac{\theta_1}{c_1}\cdot\frac{i}{k}\right)^{\frac{1}{\theta_1}}\right] \cdot a.\end{equation}
%For simplicity assume that $a=1$. Similarly to the proof of Theorem \ref{theorem:main_card}, we define $h(x) = x + \frac{1}{kc_1}[1-x]^{1-\theta_1}$ for $x\in [a_{t_0},1]$, which is increasing in the interval $I:= \left\{x: \ a_{t_0}\leq x\leq  1-\left(\frac{1-\theta_1}{kc_1}\right)^{1/\theta_1} \right\}$. Now, define
%\[b_i = 1 - \left(\left(1-\frac{\theta_0t_0}{c_0k}\right)^{\frac{\theta_1}{\theta_0}} + \frac{\theta_1t_0}{c_1k}-\frac{\theta_1}{c_1}\cdot\frac{i}{k}\right)^{\frac{1}{\theta_1}}.\]
%It is not hard to prove that for any $i\in\{t_0+1,\ldots,k\}$ we have $b_i\in I$. Then, the proof follows equivalently to the proof of Theorem \ref{theorem:main_card}. For a detailed version of the rest of the proof, we refer to the Appendix \ref{sec:proof_theorem_switch}.
 Similarly to the proof of Theorem \ref{theorem:main_card}, we define \(h(x) := x + \frac{1}{kc_{i}}[1-x]^{1-\theta_{i}}\) for $x\in [0,1]$, which is increasing in the interval \(I:= \left\{x: \ 0\leq x\leq  1-\left(\frac{1-\theta_{i}}{kc_{i}}\right)^{1/\theta_{i}} \right\}.\) %Now, define
%\[b_i = 1 - \left(\left(1-\frac{\theta_0t_0}{c_0k}\right)^{\frac{\theta_1}{\theta_0}} + \frac{\theta_1t_0}{c_1k}-\frac{\theta_1}{kc_1}\cdot i\right)^{\frac{1}{\theta_1}}.\]
%It is not hard to prove that for any $i\in\{t_0+1,\ldots,K\}$ we have $b_i\in I$. Then, the proof follows equivalently to the proof of Theorem \ref{theorem:main_card}. For a detailed version of the rest of the proof, we refer to the Appendix \ref{sec:appendix_chapter5_proofs}.
%\redTodo{FIX}
%Recall we need to show that $b_i\in I$. For this we need,
Let us prove that $b_{i}\in I$. First,  observe that $b_{i}\geq 0$. For the other inequality in $I$ we have
\begin{equation*}
b_{i} \leq 1-\left(\frac{1-\theta_{i}}{kc_{i}}\right)^{1/\theta_{i}} \quad \Leftrightarrow \quad \left(\left(1-\frac{\theta_0}{c_0k}\right)^{\theta_1/\theta_0} \ldots - \frac{\theta_{i-1}}{c_{i-1}k}\right)^{\theta_{i}/\theta_{i-1}}\geq \frac{1- \theta_{i} }{kc_{i}},
\end{equation*}
%and
%\begin{equation}\label{eq:ineq_appendix2}
%b_i \geq a_{t_0}\quad \Leftrightarrow \quad \left(1-\frac{\theta_0t_0}{c_0k}\right)^{\frac{\theta_1}{\theta_0}} + \frac{\theta_1t_0}{c_1k}-\frac{\theta_1}{c_1}\cdot\frac{i}{k} \leq \left(1-\frac{\theta_0t_0}{c_0k}\right)^{\frac{1}{\theta_0}}
%\end{equation}
%The left-hand side of inequality \eqref{eq:ineq_appendix1} is decreasing in $i$ so it is enough to prove it for $i=k$, which is
%\[k\left[c_1\left(1-\frac{\theta_0t_0}{c_0k}\right)^{\theta_1/\theta_0}-\theta_1\right] + \theta_1 t_0\geq 1-\theta_1,\]
which is satisfied for sufficiently large $k$.

%The left-hand side of inequality \eqref{eq:ineq_appendix1} is decreasing in $i$ so it is sufficient to show it for $i=t_0+1$, which is
%\[ \left(1-\frac{\theta_0t_0}{c_0k}\right)^{\frac{\theta_1}{\theta_0}} -\frac{\theta_1}{kc_1} \leq \left(1-\frac{\theta_0t_0}{c_0k}\right)^{\frac{1}{\theta_0}},\]
%which is clearly satisfied for sufficiently large $k$.

Similarly than the proof of Theorem \ref{theorem:main_card}, for $x\in [i,k]$ define
%\begin{align*}
%g(x): = & \Bigg(\bigg(\Big(1-\frac{\theta_0}{c_0k}\Big)^{\frac{\theta_1}{\theta_0}} \cdots -\frac{\theta_{i-1}}{c_{i-1}k}\bigg)^{\frac{\theta_{i}}{\theta_{i-1}}} \\
%& \qquad \qquad- \frac{\theta_{i}}{c_{i}k}\cdot (x-i)\Bigg)^{\frac{1}{\theta_{i}}}
%\end{align*}
\begin{equation*}
g(x): =  \Bigg(\bigg(\Big(1-\frac{\theta_0}{c_0k}\Big)^{\frac{\theta_1}{\theta_0}} \cdots -\frac{\theta_{i-1}}{c_{i-1}k}\bigg)^{\frac{\theta_{i}}{\theta_{i-1}}} - \frac{\theta_{i}}{c_{i}k}\cdot (x-i)\Bigg)^{\frac{1}{\theta_{i}}}
\end{equation*}
% \[g(x) :=  \left(\left(1-\frac{\theta_0t_0}{c_0k}\right)^{\frac{\theta_1}{\theta_0}} + \frac{\theta_1t_0}{c_1k}-\frac{\theta_1}{c_1k}\cdot x\right)^{\frac{1}{\theta_1}}.\]
 Note that $g'(x) = - \frac{1}{kc_{i}}\cdot g(x)^{1-\theta_{i}}$ and $g''(x) = \frac{1-\theta_{i}}{(kc_{i})^2}\cdot g(x)^{1-2\theta_{i}}$. Observe that $g$ is convex, so for any $x_1,x_2\in[i,k]$ we have \(g(x_2)\geq g(x_1) + g'(x_1)\cdot (x_2-x_1).\) By considering $x_2 = i+1$ and $x_1 = i$, we obtain
\begin{equation}\label{eq:convex_ineq}
g(i+1) - g(i) - g'(i) \geq 0
\end{equation}
Inequality \eqref{eq:convex_ineq} is exactly $h(b_{i})- b_{i+1}\geq 0$, since $g(i+1) = 1 - b_{i+1}$ and $g(i) = 1-b_i$.
%%To finish the proof, we will use induction. First, we need to prove that $a_{t_0}\geq b_{t_0}$. Let us show the following
%\begin{equation}\label{eq:ineq_appendix2}
%b_{t_0}\leq 1-\left(1-\frac{\theta_0t_0}{c_0k}\right)^{\frac{1}{\theta_0}}\ \Leftrightarrow \ \left(1-\frac{\theta_0t_0}{c_0k}\right)^{\frac{\theta_1}{\theta_0}} + \frac{\theta_1t_0}{c_1k}-\frac{\theta_1}{c_1}\cdot\frac{t_0}{k} \geq \left(1-\frac{\theta_0t_0}{c_0k}\right)^{\frac{1}{\theta_0}}
%\end{equation}
%which is true since $\theta_1\leq 1$. Inequality \eqref{eq:ineq_appendix2} implies that $a_{t_0}\geq b_{t_0}$ since $a_{t_0}\geq  1-\left(1-\frac{\theta_0t_0}{c_0k}\right)^{\frac{1}{\theta_0}}$.
Finally, since we assumed $a_{i}\geq b_{i}$, then \(a_{i+1}\geq h(a_{i}) \geq h(b_{i})\geq b_{i+1},\) where the first inequality is the definition of the recursion, the second inequality is due to the monotonicity of $h$ in the interval $I$, and finally, the last inequality was proven in \eqref{eq:convex_ineq}. Therefore, $a_k\geq b_k$ which proves the desired guarantee since $f(S^g) = a_k$.

\end{proof}

\begin{observation}
Note that we recover Theorem \ref{theorem:main_card} when $(c_i,\theta_i)=(c,\theta)$ for all $i\in\{0,\ldots, k-1\}$.
\end{observation}

%We can easily conclude Corollary \ref{corollary:one_switch} from Theorem \ref{theorem:general_approx}.
%\begin{proof}[Proof of Corollary \ref{corollary:one_switch}]
%This result easily follows by using Theorem \ref{theorem:general_approx} with $(c_i, \theta_i)= (c, \theta)$ for $i\leq t-1$ and $(c_i, \theta_i)= (c', \theta')$ for $t\leq i \leq k-1$.
%\end{proof}

\subsection{Analysis for Specific Classes of Functions}\label{sec:sharp_analysis}
Let us denote by $\mathcal{S}(f)$ the monotonic sharpness feasible region for a set function $f$, i.e., $f$ is $(c,\theta)$-monotonic sharp if, and only, if $(c,\theta)\in \mathcal{S}(f)$. We focus now on obtaining the best approximation guarantee for a monotone submodular function within region $\mathcal{S}(f)$.
\begin{proposition}\label{prop:best_guarantee}
Given a non-negative monotone submodular function $f:2^V\to\RR_+$ with monotonic sharpness region $\mathcal{S}(f)$, then the highest approximation guarantee $1-\left(1-\frac{\theta}{c}\right)^{\frac{1}{\theta}}$ for Problem \eqref{eq:problem_def} is given by a pair of parameters that lies on the boundary of $\mathcal{S}(f)$.
\end{proposition}
\begin{proof}
Fix an optimal solution $S^*$ for Problem \eqref{eq:problem_def}. Note that we can compute the best pair $(c,\theta)$ for that $S^*$ by solving the following optimization problem
\begin{align}\max & \quad 1-\left(1-\frac{\theta}{c}\right)^{\frac{1}{\theta}} \label{problem:max_guarantee}\\ \quad s.t. \ & \quad (c,\theta)\in \mathcal{S}(f,S^*),\notag\end{align}
where $\mathcal{S}(f,S^*)$ corresponds to the sharpness region with respect to $S^*$. Observe that function $1-\left(1-\frac{\theta}{c}\right)^{\frac{1}{\theta}}$ is continuous and convex in $[1,\infty)\times(0,1]$. Note that for any $c\geq 1$, if $\theta\to 0$, then  $\left(1-\frac{\theta}{c}\right)^{\frac{1}{\theta}}\to e^{-1/c}$. Also, for any subset $S$, Inequality \eqref{eq:sharp_inequality} is equivalent to \[\frac{|S^*\backslash S|}{k}\cdot\left(\frac{\sum_{e\in S^*\backslash S}f_S(e)}{\OPT}\right)^{-\theta}-c\leq 0 \] where the left-hand side is convex as a function of $c$ and $\theta$, hence $\mathcal{S}(f,S^*)$ is a convex region. Therefore, the optimal pair $(c^*,\theta^*)$ of Problem \eqref{problem:max_guarantee} lies on the boundary of $\mathcal{S}(f,S^*)$. Since we considered an arbitrary optimal set, then the result easily follows.
\end{proof}

Let us study $\mathcal{S}(f)$ for general monotone submodular functions. If we fix $|S^*\backslash S|$, the right-hand side of \eqref{eq:sharp_inequality} does not depend explicitly on $S$. On the other hand, for a fixed size $|S^*\backslash S|$, there is a subset $S^\ell$ that minimizes the left-hand side of \eqref{eq:sharp_inequality}, namely
\[\sum_{e\in S^*\backslash S}f_S(e)\geq\sum_{e\in S^*\backslash S^\ell}f_{S^\ell}(e),\]
for all feasible subset $S$ such that $|S^*\backslash S| = \ell$. For each $\ell\in[k]$, let us denote
\[W(\ell) := \sum_{e\in S^*\backslash S^\ell}f_{S^\ell}(e).\]
Therefore, instead of checking Inequality \eqref{eq:sharp_inequality} for all feasible subsets, we only need to check $k$ inequalities defined by $W(1),\ldots, W(k)$. In general, computing $W(\ell)$ is difficult since we require access to $S^*$. However, for very small instances or specific classes of functions, this computation can be done efficiently. In the following, we provide a detailed analysis of the sharpness feasible region for specific classes of functions.

% In the remaining of the paper, we will denote by $\mathcal{S}(f,S^*) = \{(c,\theta)\in[1,\infty)\times[0,1]: \ \forall S\subseteq V, \ |S|\leq k \ \text{such that} \ \eqref{eq:sharp_inequality}\}$ the feasible region in which $f$ is $(c,\theta)$-sharp with respect to $S^*$ and $\mathcal{S}(f) = \bigcup_{S^*} \mathcal{S}(f,S^*)$. In other words, $f$ is $(c,\theta)$-sharp if, and only if, $(c,\theta)\in \mathcal{S}(f)$.
%

\paragraph{Linear functions.}
Consider weights $w_e>0$ for each element $e\in V$ and function $f(S) = \sum_{e\in S}w_e$. Let us order elements by weight as follows $w_1\geq w_2\geq\ldots\geq w_n$, where element $e_i$ has weight $w_i$. We observe that an optimal set $S^*$ for Problem \eqref{eq:problem_def} is formed by the top-$k$ weighted elements and $\OPT = \sum_{i\in[k]}w_i$.
\begin{proposition}[Linear functions]\label{lemma:linear}
Consider weights $w_1\geq w_2\geq\ldots\geq w_n>0$, where element $e_i\in V$ has weight $w_i$, and denote $W(\ell) := \sum_{i=k-\ell+1}^kw_i$ for each $\ell \in \{1,\ldots,k\}$.  Then, the linear function $f(S)=\sum_{i:e_i\in S}w_i$ is $(c,\theta)$-monotonic sharp in
\begin{equation*}
\Bigg\{(c,\theta)\in[1,\infty)\times[0,1]:   \ c\geq \frac{\ell}{k}\cdot\left(\frac{W(\ell)}{W(k)}\right)^{-\theta}, \quad \forall \ \ell \in [k-1]\Bigg\}.
\end{equation*}
Moreover, this region has only $k-1$ constraints. 
\end{proposition}
\begin{proof}[Proof of Proposition \ref{lemma:linear}]
First, observe that $W(k) = \OPT$. Note that for any subset we have $|S^*\backslash S|\in\{1,\ldots,k\}$ (for $|S^*\backslash S|=0$ the sharpness inequality is trivially satisfied). Given $\ell \in [k]$, pick any feasible set $S$ such that $|S^*\backslash S| = \ell$, then the sharpness inequality corresponds to
\begin{equation}\label{eq:linear}\sum_{e\in S^*\backslash S} w_e\geq \left(\frac{\ell}{k\cdot c}\right)^{\frac{1}{\theta}}\cdot W(k),\end{equation}
where the left-hand side is due to linearity. Fix $\ell \in [k]$, we observe that the lowest possible value for the left-hand side in \eqref{eq:linear} is when $S^*\backslash S = \{e_{k-\ell+1},\ldots, e_k\}$, i.e., the $\ell$ elements in $S^*$ with the lowest weights, proving the desired result. Note that Definition \ref{def:d_sharp} is equivalent to
\begin{equation*}
\frac{W(\ell)}{W(k)}\geq \left(\frac{\ell}{k\cdot c}\right)^{\frac{1}{\theta}}, \quad\ell \in [k] \quad \Leftrightarrow \quad c\geq \frac{\ell}{k}\cdot\left(\frac{W(\ell)}{W(k)}\right)^{-\theta}, \quad\ell \in [k]
\end{equation*}
Finally, observe that $\ell = k$ is redundant with $c\geq 1$. Given this, we have $k-1$ curves that define a feasible region in which the linear function is $(c,\theta)$-monotonic sharp. In particular, if we consider $c=1$, then we can pick $\theta = \min_{\ell\in[k-1]}\left\{\frac{\log(k/\ell)}{\log(W(k)/W(\ell))}\right\}$.
\end{proof}

\noindent From Proposition \ref{lemma:linear}, we observe that the sharpness of the function depends exclusively on $w_1,\ldots,w_k$. Moreover, the weights' magnitude directly affects the sharpness parameters. Let us analyze this: assume without loss of generality that $\frac{w_k}{W(k)} \leq \frac{1}{k}$, and more generally, $\frac{W(\ell)}{W(k)} \leq \frac{\ell}{k}$ for all $\ell \in [k-1]$, so we have
\[\left(\frac{\ell}{k\cdot c}\right)^{\frac{1}{\theta}}\leq \frac{W(\ell)}{W(k)} \leq \frac{\ell}{k}, \quad\ell \in \{1,\ldots,k\}\]
This shows that a sharper linear function has more similar weights in its optimal solution, i.e., when the weights in the optimal solution are \emph{balanced}. We have the following facts for $\epsilon\in(0,1)$:
\begin{itemize}
\item If $\frac{w_k}{W(k)} = (1-\epsilon)\cdot\frac{1}{k}$, then $\frac{W(\ell)}{W(k)} \geq  (1-\epsilon)\cdot\frac{\ell}{k}$ for every $\ell\in[k-1]$. Observe that $c = \frac{1}{1-\epsilon}$ and $\theta = 1$ satisfy $(1-\epsilon)\cdot\frac{\ell}{k}\geq \left(\frac{\ell}{k\cdot c}\right)^{\frac{1}{\theta}}$ for any $\ell\in[k-1]$, showing that $f$ is $\left(\frac{1}{1-\epsilon},1\right)$-monotonic sharp. More importantly, when $\epsilon$ is small the function becomes sharper. Also, if we set $c=1$, then from the analysis of Proposition \ref{lemma:linear} we could pick
\begin{equation*}
\theta = \min_{\ell\in[k-1]}\left\{\frac{\log(k/\ell)}{\log(W(k)/W(\ell))}\right\} \geq \min_{\ell\in[k-1]}\left\{\frac{\log \frac{k}{\ell}}{\log \frac{k}{\ell (1-\epsilon)}}\right\} = \frac{\log k}{\log \frac{k}{(1-\epsilon)}},
\end{equation*}
showing that $f$ is $(1,\Omega(\frac{\log k}{\log(k/(1-\epsilon))}))$-monotonic sharp. Again, when $\epsilon\to 0$ the function becomes sharper.
\item On the other hand, suppose that $\frac{w_2}{W(k)} = \frac{\epsilon}{k}$, then $\frac{W(\ell)}{W(k)} \leq  \epsilon\cdot\frac{\ell}{k}$ for every $\ell\in[k-1]$. Similarly to the previous bullet, by setting $c=1$ we can choose
\begin{equation*}
\theta = \min_{\ell\in[k-1]}\left\{\frac{\log(k/\ell)}{\log(W(k)/W(\ell))}\right\} \leq \min_{\ell\in[k-1]}\left\{\frac{\log \frac{k}{\ell}}{\log \frac{k}{\ell\epsilon}}\right\} = \frac{\log k}{\log \frac{k}{\epsilon}},
\end{equation*}
showing that $f$ is $(1,O(\frac{\log k}{\log(k/\epsilon)}))$-sharp. Observe that when $\epsilon\to 0$ the function becomes less sharp.
\end{itemize}

\begin{observation}\label{obs:construct_linear}
Given parameters $c\geq 1$ and $\theta\in[0,1]$, it is easy to construct a linear function that is $(c,\theta)$-monotonic sharp by using Proposition \ref{lemma:linear}. Without loss of generality assume $W(k) = 1$. From constraint $\ell = 1$ choose
\(w_k = \left(\frac{1}{kc}\right)^{\frac{1}{\theta}},\)
and more generally, from constraint $\ell\in[k-1]$ choose
\(w_{k-\ell+1}= \left(\frac{\ell}{kc}\right)^{\frac{1}{\theta}} - \sum_{i=k-\ell+2}^kw_i.\) Finally, set $w_1 = 1-\sum_{i=2}^kw_i$.
\end{observation}

\begin{observation}\label{obs:construct_linear2}
Given $\beta\in[0,1]$, there exists a linear function $f$ and parameters $(c,\theta)\in[1,\infty)\times[0,1]$ such that $f$ is $(c,\theta)$-monotonic sharp and $1-\left(1-\frac{\theta}{c}\right)^{\frac{1}{\theta}}\geq 1-\beta$. To obtain this, we use Observation \ref{obs:construct_linear} with $c=1$ and any $\theta\in[0,1]$ such that $\beta\geq (1-\theta)^{1/\theta}$.
\end{observation}

\paragraph{Concave over modular.}
In this section, we will study a generalization of linear functions. Consider weights $w_e>0$ for each element $e\in V$, a parameter $\alpha\in[0,1]$ and function $f(S) = \left(\sum_{e\in S}w_e\right)^\alpha$. Observe that the linear case corresponds to $\alpha = 1$. Let us order elements by weight as follows $w_1\geq w_2\geq\ldots\geq w_n$, where element $e_i$ has weight $w_i$. Similarly than the linear case, we note that an optimal set $S^*$ for Problem \eqref{eq:problem_def} is formed by the top-$k$ weighted elements and $\OPT = \left(\sum_{i\in[k]}w_i\right)^\alpha$.
\begin{proposition}[Concave over modular functions]\label{lemma:concave}
Consider weights $w_1\geq w_2\geq\ldots\geq w_n>0$, where element $e_i\in V$ has weight $w_i$ and parameter $\alpha\in[0,1]$. Denote
\begin{equation*}
W(\ell) := \sum_{i=k-\ell+1}^k\left[\left(\sum_{j = k+1}^{k+\ell}w_{j} + \sum_{j=1}^{k-\ell}w_{j}+w_i\right)^\alpha - \left(\sum_{j = k+1}^{k+\ell}w_{j} + \sum_{j=1}^{k-\ell}w_{j}\right)^\alpha\right]
\end{equation*}
for each $\ell \in \{1,\ldots,k\}$.  Then, the function $f(S)=\left(\sum_{i:e_i\in S}w_i\right)^\alpha$ is $(c,\theta)$-monotonic sharp in
\begin{equation*}
\Bigg\{(c,\theta)\in[1,\infty)\times[0,1]:  c\geq \frac{\ell}{k}\cdot\left(\frac{W(\ell)}{\OPT}\right)^{-\theta}, \quad \forall \ \ell \in [k]\Bigg\}.
\end{equation*}
\end{proposition}

\begin{proof}
First, observe that unlike the linear case, $W(k) \neq \OPT$. Given $\ell \in [k]$, pick any feasible set such that $|S^*\backslash S| = \ell$, then the sharpness inequality corresponds to
\begin{equation}
\sum_{e\in S^*\backslash S} \left(\sum_{e'\in S}w_{e'} + w_e\right)^\alpha-\left(\sum_{e'\in S}w_{e'} \right)^\alpha \geq \left(\frac{\ell}{k\cdot c}\right)^{\frac{1}{\theta}}\cdot \OPT.\label{eq:concave}
\end{equation}
Observe that function $(x+y)^\alpha -x^\alpha$ is increasing in $y$ and decreasing in $x$. Therefore, the lowest possible value for the left-hand side in \eqref{eq:concave} is when $\sum_{e'\in S}w_{e'}$ is maximized and $w_e$ is minimized. Given this, for each $\ell \in [k]$ we choose $S=\{e_1,\ldots,e_{k-\ell},e_{k+1},\ldots,e_{k+\ell}\}$. In this way, we get $S^*\backslash S = \{e_{k-\ell+1},\ldots, e_k\}$, whose elements have the lowest weights, and $S$ has the highest weight possible in $V\backslash\{e_{k-\ell+1},\ldots, e_k\}$. Hence, Definition \ref{def:d_sharp} is equivalent to
\begin{equation*}
\frac{W(\ell)}{\OPT}\geq \left(\frac{\ell}{k\cdot c}\right)^{\frac{1}{\theta}}, \quad\ell \in [k] \quad \Leftrightarrow \quad c\geq \frac{\ell}{k}\cdot\left(\frac{W(\ell)}{\OPT}\right)^{-\theta}, \quad\ell \in [k].
\end{equation*}
We have $k$ curves that define a feasible region in which the function is $(c,\theta)$-monotonic sharp with respect to $S^*$. In particular, if we consider $c=1$, then we can pick $\theta = \min_{\ell\in[k]}\left\{\frac{\log(k/\ell)}{\log(\OPT/W(\ell))}\right\}$.
\end{proof}

%\subsubsection*{Facility-location functions}
%Consider $V=\{1,\ldots,n\}$, weights $w_{ij}>0$ for each pair $(i,j)\in V\times V$ and function $f(S) = \sum_{i\in[n]}\max_{j\in S}w_{ij}$. For each $i\in[n]$, denote by $w_{ij^*} = \argmax_{j\in[n]}w_{ij}$ the largest weight. Let us order those weights as follows $w_{1j^*_1}\geq w_{1j^*_2}\geq \cdots\geq w_{nj_n^*}$.  We observe that an optimal set $S^*$ for Problem \ref{eq:problem_def} corresponds to $\{{j_1^*},\ldots, {j_k^*}\}$. In particular, we obtain the following result for this type of functions
%\begin{lemma}[Facility-location functions]\label{lemma:linear}
%Consider weights $w_{ij}>0$ for each pair $(i,j)\in V\times V$. Then, the function $f(S) = \sum_{i\in[n]}\max_{j\in S}w_{ij}$ is $(c,\theta)$-sharp only for parameters $c$ and $\theta$ such that $(1/c)^{1/\theta}\to 0$.
%\end{lemma}
%\begin{proof}
%Observe that for this type of functions, Inequality \eqref{eq:sharp_inequality} corresponds to
%\begin{equation}\label{eq:facility_loc}
%\sum_{e\in S^*\backslash S}\sum_{i\in[n]}\left(\max\left\{\max_{j\in S}w_{ij},w_{i,e}\right\} - \max_{j\in S}w_{ij} \right)\geq\left(\frac{|S^*\backslash S|}{k\cdot c}\right)^{\frac{1}{\theta}}\cdot \OPT
%\end{equation}
%Recall that $S^*=\{{j_1^*},\ldots, {j_k^*}\}$, so consider $S =\{{j_1^*},\ldots, {j_{k-1}^*}\}$ in \eqref{eq:facility_loc}
%\end{proof}

\paragraph{Coverage function, \citep{nemhauser1978hardness}.}

Consider a space of points $\mathcal{X}=\{1,\ldots,k\}^k$, sets $A_i = \{x\in\mathcal{X}: \ x_i = 1\}$ for $i\in[k-1]$ and $B_i=\{x\in\mathcal{X}: \ x_k = i\}$ for $i \in [k]$, ground set $V = \{A_1,\ldots, A_{k-1}, B_1,\ldots, B_k\}$, and function $f(S) = \left|\bigcup_{U\in S}U\right|$ for $S\subseteq V$. In this case, Problem \eqref{eq:problem_def} corresponds to finding a family of $k$ elements in $V$ that maximizes the coverage of $\mathcal{X}$. By simply counting, we can see that the optimal solution for Problem \eqref{eq:problem_def} is $S^* = \{B_1,\ldots, B_k\}$ and $\OPT = k^k$. As shown in \citep{nemhauser1978hardness}, the greedy algorithm achieves the best possible $1-1/e$ guarantee for this problem.
\begin{proposition}\label{prop:worst_case}
Consider ground set $V = \{A_1,\ldots, A_{k-1}, B_1,\ldots, B_k\}$. Then, the function $f(S) = \left|\bigcup_{U\in S}U\right|$ is $(c,\theta)$-monotonic sharp in
\begin{equation*}
\Bigg\{(c,\theta)\in[1,\infty)\times[0,1]: \ c\geq \frac{\ell}{k}\cdot\left(\frac{\ell}{k}\cdot\left(\frac{k-1}{k}\right)^\ell\right)^{-\theta}, \quad \forall \ \ell \in [k-1]\Bigg\}.
\end{equation*}
\end{proposition}
\begin{proof}
First, note that any family of the form $\{A_{i_1},\ldots, A_{i_\ell}, B_{j_1},\ldots, B_{j_{k-\ell}}\}$ covers the same number of points for $\ell\in[k-1]$. Second, since there are only $k-1$ sets $A_i$, then any subset $S\subseteq V$ of size $k$ satisfies $|S^*\backslash S|\leq k-1$. By simply counting, for $\ell \in [k-1]$ and set $S$ such that $|S^*\backslash S| = \ell$, we have
\begin{align*}
f(S) &= k^k - \ell k^{k-\ell - 1} (k-1)^\ell,\\
f(S+e) & =   k^k - (\ell -1)k^{k-\ell - 1} (k-1)^\ell.
\end{align*}
Then, \eqref{eq:sharp_inequality} becomes
\[\frac{\ell}{k}\cdot\left(\frac{k-1}{k}\right)^\ell\geq \left(\frac{\ell}{kc}\right)^{\frac{1}{\theta}}\]
Observe that $f$ is $\left(1,\frac{1}{k}\right)$-sharp since $\ell \leq k-1$ and $\left(\frac{k-1}{k}\right)^\ell\geq \left(\frac{\ell}{k}\right)^{k-1}$.
\end{proof}

\begin{observation}[Coverage function, \citep{nemhauser1978hardness}.]\label{obs:worst_case}
Note that in order to achieve $1-\left(1-\frac{\theta}{c}\right)^{\frac{1}{\theta}}\geq 1-e^{-1}$, we need $\theta\in[0,1]$ and $1\leq  c\leq \frac{\theta}{1-e^{-\theta}}$. On the other hand, by taking $\ell = k-1$ in Proposition \ref{prop:worst_case} we have $c\geq \left(\frac{k-1}{k}\right)^{-k\theta +1}$, where the right-hand side tends to $e^\theta$ when $k$ is sufficiently large. Therefore, for $k$ sufficiently large we have $e^\theta\leq c \leq \frac{\theta}{1-e^{-\theta}}$, whose only feasible point is $c= 1$ and $\theta\to 0$. %This shows that for $k$ sufficiently large, the coverage function given in \citep{nemhauser1978hardness} has a sharpness feasible region $\mathcal{S}(f) = \{(c,\theta)\in[1,\infty)\times[0,1]: \ (1/c)^{1/\theta}\to 0\}$.
\end{observation}

\section{Analysis of Submodular Sharpness}\label{sec:s_sharpness}

In this section, we focus on the analysis of the standard greedy algorithm for Problem \eqref{eq:problem_def} when the objective function is $(c,\theta)$-submodular sharp. First, let us prove the following basic facts:
\begin{lemma}\label{lemma:subm_sharpness_facts}
Consider any monotone submodular set function $f:2^V\to\RR_+$. Then,
\begin{enumerate}
\item There is always a set of parameters $c$ and $\theta$ such that $f$ is $(c,\theta)$-submodular sharp. In particular, $f$ is always $(c,\theta)$-submodular sharp when both $c\to1$ and $\theta\to0$.
\item If $f$ is $(c,\theta)$-submodular sharp, then for any $c'\geq c$ and $\theta'\leq \theta$, $f$ is $(c',\theta')$-submodular sharp. Therefore, in order to maximize the guarantee of Theorem \ref{theorem:total_sharp} we look for the smallest feasible $c$ and the largest feasible $\theta$.
\item Definition \ref{def:d_sharp} is stronger than Definition \ref{def:total_d_sharp}.
%\item Inequality \eqref{eq:sharp_inequality} needs to be checked only for sets of size exactly $k$.
\end{enumerate}
\end{lemma}
\begin{proof}
\begin{enumerate}
\item Note that $f$ satisfies the following sequence of inequalities for any set $S$:
\begin{equation}\label{eq:lemma_subm_sharp1}
\max_{e\in S^*\backslash S}f_S(e) \geq \frac{\sum_{e\in S^*\backslash S}f_S(e)}{|S^*\backslash S|} \geq  \frac{f(S\cup S^*) - f(S)}{k} \geq  \frac{f(S^*) - f(S)}{k}
\end{equation}
where the second inequality is because of submodularity and in the last inequality we applied monotonicity. Observe that \eqref{eq:lemma_subm_sharp1} is exactly \eqref{eq:total_sharp_inequality} for $c=1$ and $\theta \to 0$.

\item Observe that $\frac{[f(S^*) -f(S)]^{1-\theta} f(S^*)^\theta}{k\cdot c}$ as a function of $c$ and $\theta$ is increasing in $\theta$ and decreasing in $c$. Therefore, $\frac{[f(S^*) -f(S)]^{1-\theta} f(S^*)^\theta}{k\cdot c}\geq\frac{[f(S^*) -f(S)]^{1-\theta'} f(S^*)^{\theta'}}{k\cdot c'}$ for  $c'\geq c$ and $\theta'\leq \theta$.
\item Definition \ref{def:d_sharp} implies that
\begin{equation}\label{eq:lemma_subm_sharp2}
\frac{\sum_{e\in S^*\backslash S}f_S(e)}{|S^*\backslash S|} \geq \frac{\left[\sum_{e\in S^*\backslash S}f_S(e)\right]^{1-\theta}f(S^*)^\theta}{kc}.
\end{equation}
%From \citep{nemhauser1978analysis} and the definition of curvature we know that
%\[f(S\cup S^*) - f(S^*) \geq \sum_{e\in S\backslash S^*} f_{S\cup S^* - e}(e) \geq (1-\gamma)\sum_{e\in S\backslash S^*}f(e).\]
On the other hand, by using submodularity and monotonicity we get
\[ \sum_{e\in S^*\backslash S}f_S(e)\geq f(S^*\cup S)-f(S)\geq f(S^*) - f(S).\]
Therefore, by using \eqref{eq:lemma_subm_sharp2} we obtain
%\begin{equation}\label{eq:lemma_subm_sharp3} \sum_{e\in S^*\backslash S}f_S(e)\geq f(S^*)-f(S)+(1-\gamma)\sum_{e\in S\backslash S^*}f(e).\end{equation}
%
%By using \eqref{eq:lemma_subm_sharp2} and \eqref{eq:lemma_subm_sharp3}, we get
\[\max_{e\in S^*\backslash S}f_S(e)\geq  \frac{\sum_{e\in S^*\backslash S}f_S(e)}{|S^*\backslash S|} \geq \frac{\left[f(S^*)-f(S)\right]^{1-\theta}f(S^*)^\theta}{kc},\]
which proves the desired result.
%\item Given that $f$ is submodular, then $f_A(e)\geq f_B(e)$ for any $e\in V$ and $A\subseteq B\subseteq V\backslash e$, so in \eqref{eq:sharp_inequality} we can assume without loss of generality that $|S| = k$.
\end{enumerate}
\end{proof}

\begin{proof}[Proof of Theorem \ref{theorem:total_sharp}]
Let us denote by $S_{i}:=\{e_1,\ldots,e_i\}$ the set we obtain in the $i$-th iteration of Algorithm \ref{alg:greedy} and $S_0 = \emptyset$. Note that $S^g:= S_k$. Since the greedy algorithm chooses the element with the largest marginal in each iteration, then for all $i\in[k]$ we have
\[f(S_i) - f(S_{i-1})\geq \max_{e\in S^*\backslash S_{i-1}}f_{S_{i-1}(e)}\]
Now, from the submodular sharpness condition we conclude that
\begin{equation} \label{eq:theorem_subm_sharp1} f(S_i) - f(S_{i-1})\geq  \frac{\left[f(S^*)-f(S_{i-1})\right]^{1-\theta}f(S^*)^\theta}{kc}.\end{equation}
The rest of the proof is the same as the proof of Theorem \ref{theorem:main_card}, which gives us the desired result.
\end{proof}

Finally, we need to prove the main result for the concept of dynamic submodular sharpness. This proof is similar to the proof of Theorem \ref{theorem:general_approx}.

\begin{proof}[Proof of Theorem \ref{theorem:general_approx2}]
For each iteration $i\in[k]$ in the greedy algorithm we have
\[f(S_i) - f(S_{i-1})\geq \max_{e\in S^*\backslash S_{i-1}}f_{S_{i-1}(e)}\]
Now, from the dynamic submodular sharpness condition we conclude that
\begin{equation} \label{eq:theorem_subm_sharp1} f(S_i) - f(S_{i-1})\geq  \frac{\left[f(S^*)-f(S_{i-1})\right]^{1-\theta_{i-1}}f(S^*)^{\theta_{i-1}}}{kc_{i-1}},\end{equation}
which gives the same recurrence than Theorem \ref{theorem:general_approx}. The rest of the proof is the same as the proof of Theorem \ref{theorem:general_approx}.

\end{proof}

\section{Experiments}\label{sec:experiments}

In this section, we provide a computational study of the sharpness criteria in two real-world applications: movie recommendation and non-parametric learning. In these experiments, we aim to explicitly obtain the sharpness parameters of the objective function for different small ground sets. With these results, we will empirically show how the approximation factors vary with respect to different instances defined by the cardinality budget $k$. We will observe that the curvature analysis \citep{conforti_etal84}, submodular stability \citep{chatzia_etal17} and monotonic sharpness are not enough, but more refined concepts as dynamic monotonic sharpness and submodular sharpness in its two versions provide strictly better results.

\paragraph{Search for monotonic sharpness.} Fix an optimal solution $S^*$. For each $\ell \in [k]$, compute 
\[W(\ell) :=\min_S\left\{ \sum_{e\in S^*\backslash S}f_{S}(e): |S|\leq k, \ |S^*\backslash S| = \ell\right\}.\]
To find parameters $(c,\theta)$ we follow a simple search: we sequentially iterate over possible values of $c$ in a fixed range $[1,c_{max}]$ (we consider granularity $0.01$ and $c_{max} = 3$). Given $c$, we compute \(\theta = \min_{\ell\in[k]}\left\{\frac{\log(kc/\ell)}{\log(\OPT/W(\ell))}\right\}.\) Once we have $c$ and $\theta$, we compute the corresponding approximation factor. If we improve we continue and update $c$; otherwise, we stop. A similar procedure is done for the case of dynamic sharpness.

\paragraph{Search for submodular sharpness.} Fix an optimal solution $S^*$. To find parameters $(c,\theta)$ we follow a simple search: we sequentially iterate over possible values of $c$ in a fixed range $[1,c_{max}]$ (we consider granularity $0.01$ and $c_{max} = 3$). Given $c$, we compute \(\theta = \min_{|S|\leq k}\left\{\frac{\log(kcW_2(S)/W(S))}{\log(\OPT/W(S))}\right\},\) where
\[W(S) := \OPT - f(S) \quad \text{and} \quad W_2(S) := \max_{e\in S^*\backslash S}f_S(e). \]
 Once we have $c$ and $\theta$, we compute the corresponding approximation factor. If we improve we continue and update $c$; otherwise, we stop. A similar procedure is done for the case of dynamic submodular sharpness.
 
 \paragraph{Experiments setup.} For each application, we will run experiments on small ground sets. For each budget size $k$, we sample $n=2k$ elements from the data sets which will be considered as the ground set. In each graph, we will plot the approximation factor (y-axis) obtained by the corresponding method in each instance (x-axis). The analysis we will study are: curvature (defined in Section \ref{sec:related_work}), monotonic and submodular sharpness (computed as described above), and finally, the greedy ratio (worst possible value output by the greedy algorithm in the corresponding instance).

\subsection{Non-parametric Learning}
For this application we follow the setup in \citep{mirzasoleiman_etal15}. Let $X_V$ be a set of random variables corresponding to bio-medical measurements, indexed by a ground set of patients $V$. We assume $X_V$ to be a Gaussian Process (GP), i.e., for every subset $S\subseteq V$, $X_S$ is distributed according to a multivariate normal distribution $\mathcal{N}(\muB_S,\SigmaB_{S,S})$, where $\muB_S = (\mu_{e})_{e\in S}$ and $\SigmaB_{S,S} = [\K_{e,e'}]_{e,e'\in S}$ are the prior mean vector and prior covariance matrix, respectively. The covariance matrix is given in terms of a positive definite kernel $\K$, e.g., a common choice in practice is the squared exponential kernel $\K_{e,e'} =\exp(-\|x_e-x_{e'}\|^2_2/h)$. Most efficient approaches for making predictions in GPs rely on choosing a small subset of data points. For instance, in the Informative Vector Machine (IVM) the goal is to obtain a subset $A$ such that maximizes the information gain, \(f(A) = \frac{1}{2}\log\text{det}(\ident + \sigma^{-2}\SigmaB_{A,A})\), which was shown to be monotone and submodular in \citep{krause_guestrin05}.

In our experiment, we use the Parkinson Telemonitoring dataset \citep{tsanas_etal10} consisting of a total of $5,875$ patients with early-stage Parkinson's disease and the corresponding bio-medical voice measurements with 22 attributes (dimension of the observations). We normalize the vectors to zero mean and unit norm. With these measurements we computed the covariance matrix $\Sigma$ considering the squared exponential kernel with parameter $h=0.75$. For the objective function we consider $\sigma = 1$. As we mentioned before, the objective in this application is to select the $k$ most informative patients. 

The objective function in this experiment is highly non-linear which makes difficult to obtain the sharpness parameters. Therefore, for this experiment we consider different small random instances with $n=2k$ where $k=\{5,\ldots,10\}$.  In Figure \ref{fig:approx_factors} (a) we plot the variation of the approximation factors with respect to different instances of size $n=2k$. Observe that the greedy algorithm finds a nearly optimal solution in each instance. The best approximation factor is obtained when using the concept of dynamic submodular sharpness, which is considerably close to the greedy results. These results significantly improve the ones obtained by the curvature analysis and monotonic sharpness, providing evidence that more refined notions of sharpness can capture the behavior of the greedy algorithm. 

\subsection{Movie Recommendation}
For this application we consider the MovieLens data-set \citep{movielens} which consists of 7,000 users and 13,977 movies. Each user had to rank at least one movie with an integer value in $\{0,\ldots, 5\}$ where 0 denotes that the movies was not ranked by that user. Therefore, we have a matrix $[r_{ij}]$ of rankings for each user $i$ and each movie $j$. The objective in this application is to select the $k$ highest ranked movies among the users. To make the computations less costly in terms of time, we use only $m = 1000$ users. In the same spirit, we will choose a small number $n$ from the 13,977 movies.

In our first experiment, we consider the following function \(f(S) = \left(\frac{1}{m}\sum_{i\in[m]}\sum_{j\in S} r_{ij}\right)^\alpha\) where $\alpha\in(0,1]$. We consider $\alpha = 0.8$ and different small random instances with $n=2k$ where $k=\{5,\ldots,10\}$. First, we noticed in our experiment that the instance is not submodular stable \citep{chatzia_etal17} since it had multiple optimal solutions. In Figure \ref{fig:approx_factors} (b) we plot the variation of the approximation factors with respect to different $k$'s. We observe that monotonic sharpness already gives us improved guarantees with respect to the worst-case $1-1/e$, although worse results than the curvature analysis. More refined definitions as the submodular sharpness slightly improve the results for any instance.% We also plot in blue the best  approximation factor that can be obtained by using a single switching moment. The results for this approach are significantly close to the general dynamic sharpness.

% random instance of $n = 2,000$ movies and we want to select $k=1,000$ movies. We take 10 different values for $\alpha\in\{0.1,0.2,\ldots, 1.0\}$. First, we noticed in our experiment that the instance is not submodular stable \citep{chatzia_etal17} since it had multiple optimal solutions. In Figure \ref{fig:aaai2020} (b) we plot the variation of the approximation factors with respect to different values of $\alpha$. We observe that there is an $\alpha^*$ in which the curvature analysis improves the sharpness criterion. Below this $\alpha^*$, the dynamic sharpness with $k-1$ switching points (magenta) clearly outperforms the curvature analysis; similarly, for the static sharpness. We conclude that the curvature analysis is ideal when the objective function is closer to be a linear function.

%In the next experiment, we fix $\alpha =0.8$ and we consider different random instances with $n=2k$, where $k=\{50,100,200,500,800,1000\}$. 
In the next experiment, we consider the facility-location function $f(S) = \frac{1}{m}\sum_{i\in[m]}\max_{j\in S} r_{ij}$. This function is known to be non-negative, monotone, and submodular. Most of the time this function is not 2-stable \citep{chatzia_etal17} since it has multiple optimal solutions. For this function, we consider different small random instances with $n=2k$ elements in the ground set where $k=\{5,\ldots,10\}$. In Figure \ref{fig:approx_factors} (c) we plot the variation of the approximation factors with respect to different values of $k$. We observed that the greedy algorithm (orange) always finds an optimal solution. We note that the curvature analysis and monotonic sharpness barely improve the worst-case ratio $1-1/e$. We obtain a significant improvement if we use the submodular sharpness approach. However, the gap between the greedy results and the dynamic submodular sharpness is still substantial, which may be due to the shape of this objective function: facility-location functions are known to be \emph{flat} and have multiple optimal solutions.

\subsection{Exemplar-based Clustering}\label{sec:image_clustering}

We follow the setup in \citep{mirzasoleiman_etal15}. Solving the \(k\)-medoid problem is a common way to select a subset of exemplars that represent a large dataset $V$ \citep{kaufman_09}. This is done by minimizing the sum of pairwise dissimilarities between elements in $A\subseteq V$ and $V$. Formally, define \(L(A) =\frac{1}{V}\sum_{e\in V} \min_{v\in A}d(e,v)\), where \(d: V\times V\to\RR_+\) is a distance function that represents the dissimilarity between a pair of elements. By introducing an appropriate auxiliary element $e_0$, it is possible to define a new objective \(f(A) := L(\{e_0\}) - L(A+e_0)\) that is monotone and submodular \citep{gomes_krause10}, thus maximizing \(f\) is equivalent to minimizing $L$.
%In our experiment, we use a subset of the images in the VOC2012 dataset \citep{voc_2012}. The ground set $V$ corresponds to images, and we want to select a subset of them that most represent the dataset. Each image has several (maybe repeated) associated categories. Each category indicates a certain visual queue appearing in the image such as chair, bird, hand, etc. Therefore, images are represented by feature vectors obtained from categories (around 20). For example, if there were two categories $a$ and $b$, and an image had features $[a,a,b]$, then its feature vector would be $(2,1)$. We consider the euclidean distance $d(e,e') = \|x_e-x_{e'}\|$ where $x_e,x_{e'}$ are the feature vectors for images $e,e'$. We study two type of instances: small with $n=5,000$ and large $n=17,000$. In both of them, we take $k=20$ perturbations of the function $f$ defined in the paragraph above, i.e., problem \eqref{eq:problem_def} corresponds to \(\max_{A\in \I}\min_{i\in[20]}f(A) + \sum_{e\in A\cap \Lambda_i}\eta_e\), where \(\Lambda_i \) is a random set different for each $i\in[20]$ of size 1,000 for the small instance and 3,000 for the large one, and finally, $\eta\sim[0,1]^V$ is a uniform error vector. For the small instance, we made 8 runs where in each run we change the number of parts $q\in\{3,\ldots, 10\}$, and we consider a random budget $b$ uniformly selected from $\{5,\ldots,9\}$. For the large instance, we consider $q=3$ parts and budget $b=20$.
In our experiment, we use the VOC2012 dataset \citep{voc_2012} which contains around 10,000 images. The ground set $V$ corresponds to images, and we want to select a subset of the images that best represents the dataset. Each image has several (possible repeated) associated categories such as person, plane, etc. There are around 20 categories in total. Therefore, images are represented by feature vectors obtained by counting the number of elements that belong to each category, for example, if an image has 2 people and one plane, then its feature vector is $(2,1,0,\ldots,0)$ (where zeros correspond to other elements). We choose the Euclidean distance $d(e,e') = \|x_e-x_{e'}\|$ where $x_e,x_{e'}$ are the feature vectors for images $e,e'$. We normalize the feature vectors to mean zero and unit norm, and we choose $e_0$ as the origin.

For this experiment, we consider different random small instances with $n=2k$ where $k=\{5,\ldots,10\}$. The objective function in this experiment is non-linear which makes difficult to obtain the sharpness parameters. In Figure \ref{fig:approx_factors} (d)  we plot the variation of the approximation factors with respect to different values of $k$.  The  dynamic submodular sharpness approach outperforms the rest of the procedures, although the greedy algorithm always finds an optimal solution.

\begin{figure*}[htb]
\centering
\begin{tabular}{cc}
\includegraphics[width=0.48\textwidth]{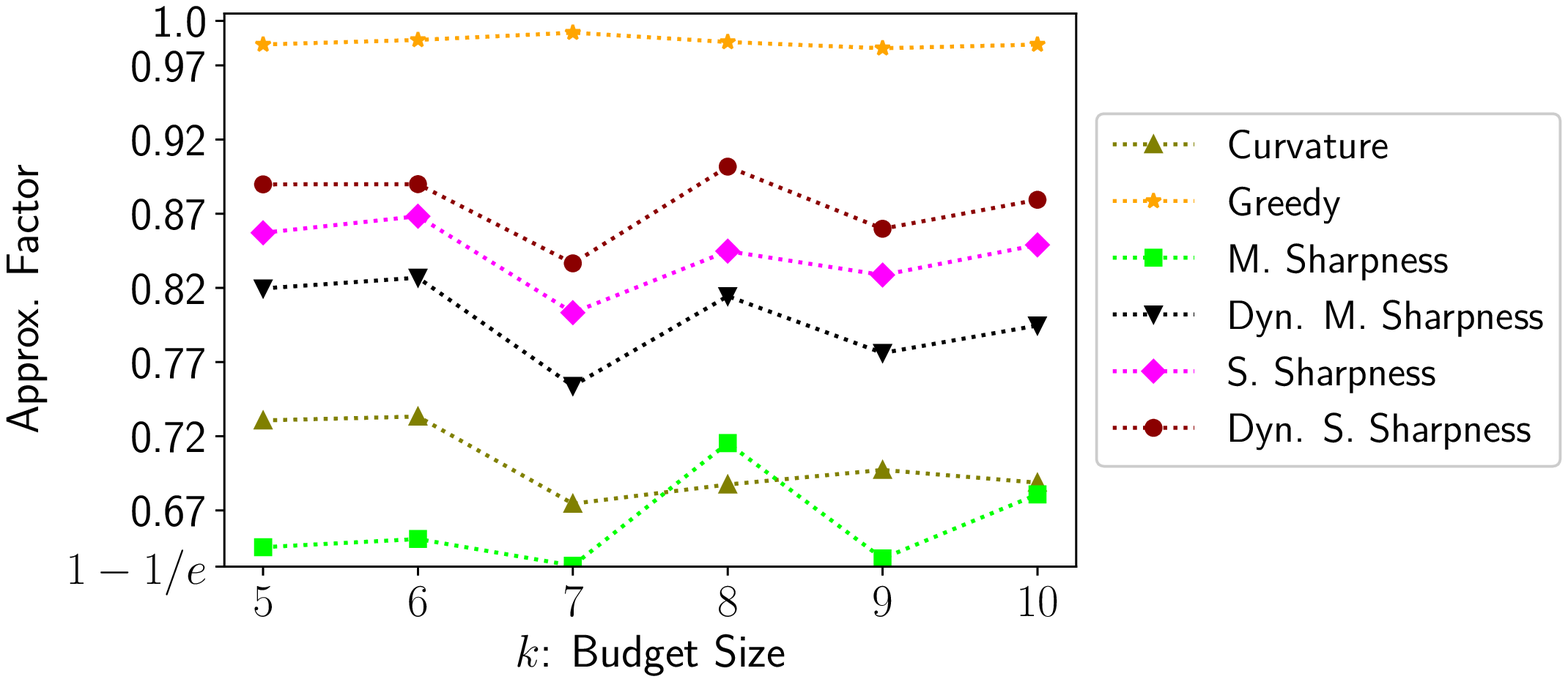}&
\includegraphics[width=0.48\textwidth]{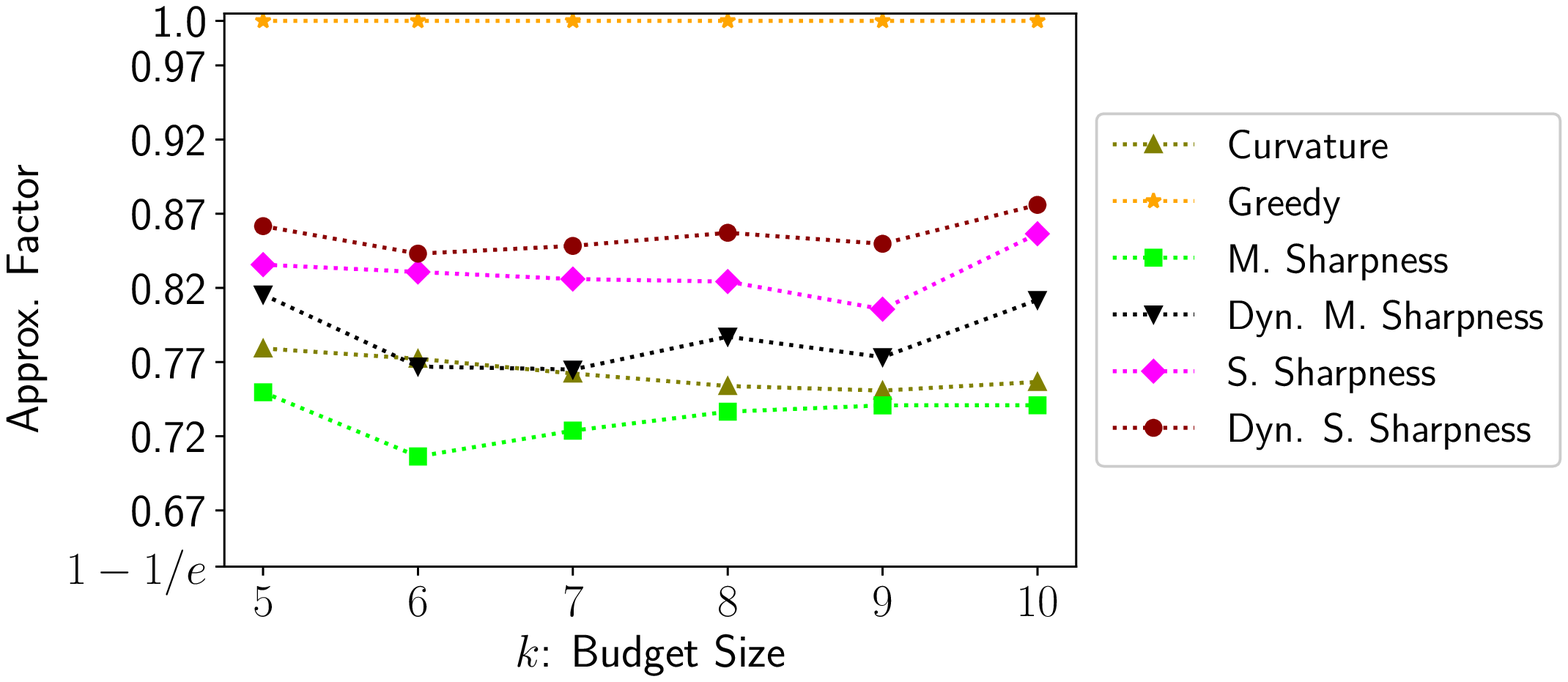}\\
(a) & (b)\\
\includegraphics[width=0.48\textwidth]{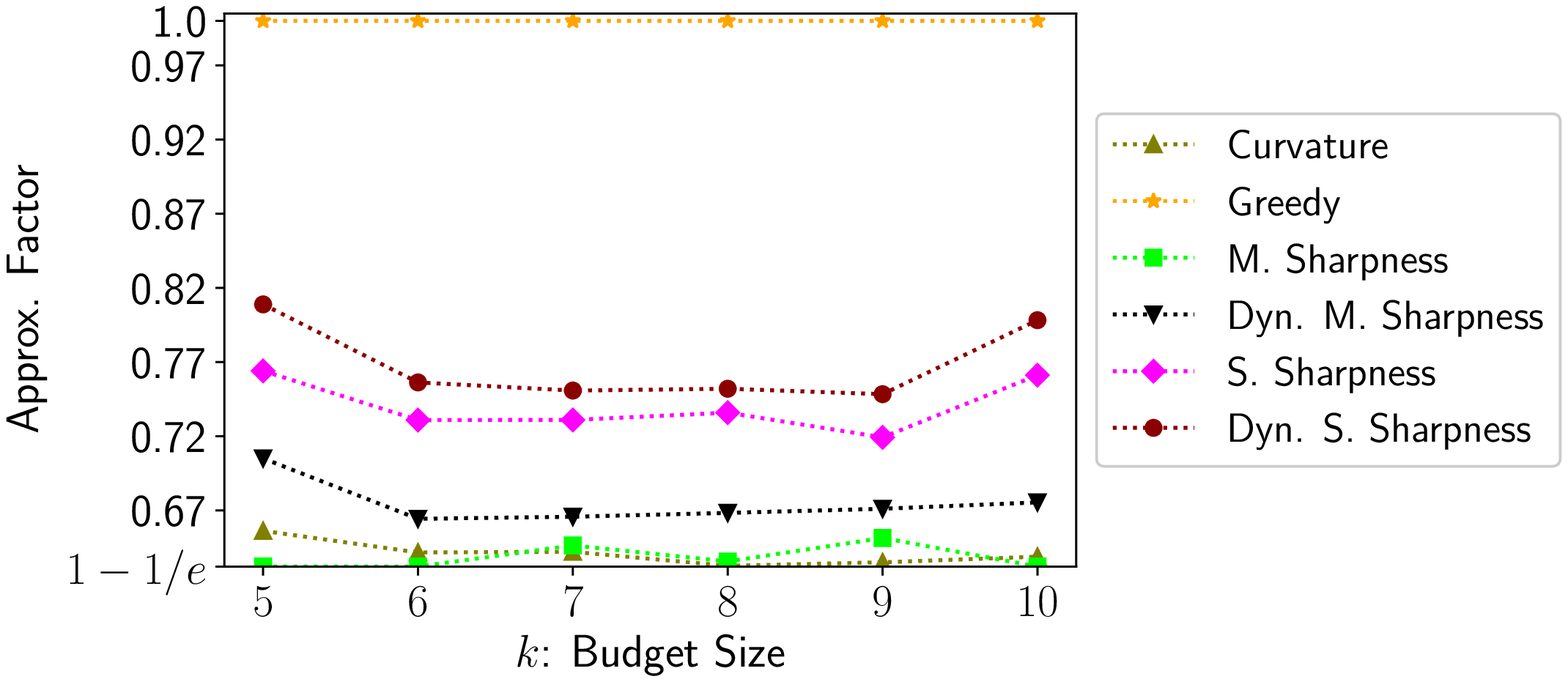}  &
\includegraphics[width=0.48\textwidth]{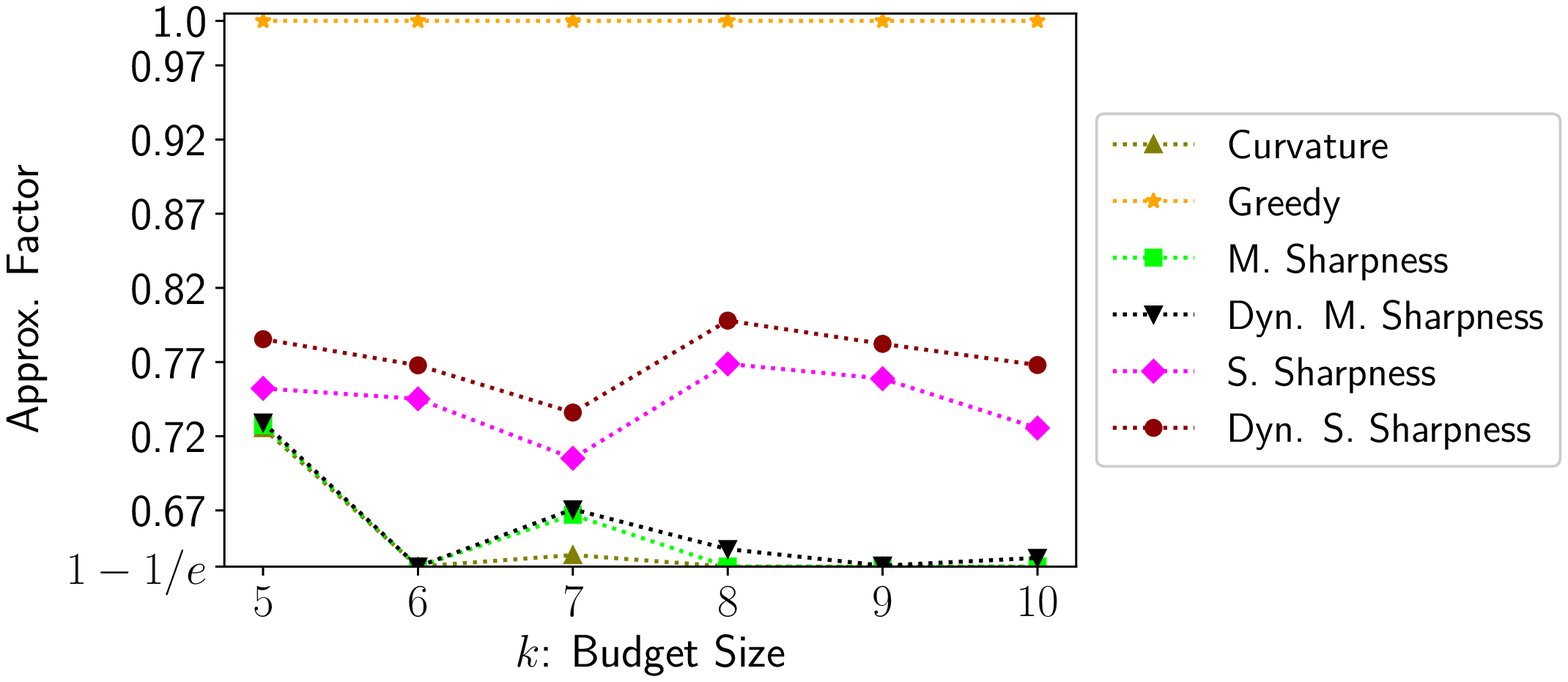}  \\
 (c) & (d) 
 \end{tabular}
% \begin{tabular}{cc}
%\includegraphics[width=0.48\textwidth]{Figures2/parkinson_approx_factor_N.eps}&
%\includegraphics[width=0.48\textwidth]{Figures2/sqrt_small_approx_factor_N.eps}\\
%(a) & (b)\\
%\includegraphics[width=0.48\textwidth]{Figures2/fac_loc_approx_factor_N.eps}  &
%\includegraphics[width=0.48\textwidth]{Figures2/image_clustering_approx_factor_N.eps}  \\
% (c) & (d) 
% \end{tabular}
 \caption{ Variations in the approximation factors with respect to different budgets $k$. {\it Non-parametric learning:} (a); {\it Movie recommendation:} (b) concave over modular with $\alpha =0.8$ and (c) facility location function; {\it Exemplar-based clustering:} (d)}
\label{fig:approx_factors}
\end{figure*}

\section{Extensions}

%\subsection{Basic Facts about Sharpness}

%\subsection{Remaining Proofs}
%
%
%
%\begin{proof}[Proof of Theorem \ref{theorem:general_approx}]
%
%\end{proof}

\subsection{Approximate Sharpness}\label{sec:approximate_sharpness}
Consider a parameter $\delta\in\left[0,1-\frac{1}{k}\right]$. For a set function $f:2^V\to\RR_+$ we denote the \emph{$\delta$-marginal value} for any subset $A\subseteq V$ and $e\in V$ by $f^\delta_A(e):=f(A+e)-(1-\delta)f(A)$. Now, for a given set of parameters $\delta\in\left[0,1-\frac{1}{k}\right]$, $\theta\in[0,1]$, and $c\geq 1$, we define \emph{approximate submodular sharpness} as follows,

\begin{definition}[Approximate Submodular Sharpness]\label{def:d_delta_sharp}
A non-negative monotone submodular function $f:2^V\to\RR_+$ is said to be \emph{$(\delta,c,\theta)$-sharp}, if there exists an optimal solution $S^*$ such that for any subset $S\subseteq V$ with $|S|\leq k$ the function satisfies
\begin{align*}
\sum_{e\in S^*\backslash S} f_S^\delta(e)\geq \left[\frac{|S^*\backslash S|}{k\cdot c}\right]^{\frac{1}{\theta}}\cdot f(S^*)
\end{align*}
\end{definition}

When relaxing the notion of marginal values to $\delta$-marginal values, we get the following result

\begin{theorem}\label{theorem:main_card2}
Consider a non-negative monotone submodular function $f:2^V\to\RR_+$ that is $(\delta,c,\theta)$-sharp. Then, for problem \eqref{eq:problem_def} the greedy algorithm returns a feasible set $S^g$ such that
\[f(S^g)\geq \frac{1}{1-\delta+\delta kc}\left[1-(1-\delta)^{\frac{1}{\theta}}\left(1-\frac{\theta}{c}\right)^{\frac{1}{\theta}}\right] \cdot f(S^*).\]
\end{theorem}

\noindent Observe that for $\delta =0$, we recover the result in Theorem \ref{theorem:main_card}.
%We observe from real-world data that often the left-hand side in Definition  \ref{def:d_sharp} decreases significantly for subsets that do not intersect the optimal solution. This translates in sharpness parameters close to the extreme case of an arbitrary submodular function, i.e., $c \approx 1$ and $\theta \approx 0$.  However, in experiments the greedy algorithm still performs considerably well for non-linear objectives. This motivates us to define approximate submodular sharpness in Defintion \ref{def:d_delta_sharp}. Recall that given parameters $c\geq 1$, $\delta\in[0,1-1/k]$ and $\theta\in [0,1]$, a function is $(\delta,c,\theta)$-sharp if there exists an optimal set $S^*$ such that for any set $S$ with at most $k$ elements, then
%\[\sum_{e\in S^*\backslash S} f^\delta_S(e)\geq \left[\frac{|S^*\backslash S|}{k\cdot c}\right]^{\frac{1}{\theta}}\cdot f(S^*)
%\]
%
%With the previous definition in mind, we are ready to analyze the performance of the standard greedy algorithm for problem \eqref{eq:problem_def} when the objective function is $(\delta,c,\theta)$-sharp.
In the following proof we will observe that the greedy algorithm automatically adapts to parameters $\delta$, $\theta$, and $c$, and they will not be required to be part of the input.
\begin{proof}[Proof of Theorem \ref{theorem:main_card2}]
Similarly to the proof of Theorem \ref{theorem:main_card}, let us denote by $S_{i}$ the set we obtained in the $i$-th iteration of the greedy algorithm. By using the properties of the function $f$, we obtain the following sequence of inequalities
\begin{align}
f(S_i) &- (1-\delta)f(S_{i-1}) = \notag\\
&\frac{\sum_{e\in S^*\backslash S_{i-1}}f(S_i) - (1-\delta)f(S_{i-1})}{|S^*\backslash S_{i-1}|} \notag\\
&\geq \frac{\sum_{e\in S^*\backslash S_{i-1}}f(S_{i-1}+e) - (1-\delta)f(S_{i-1})}{|S^*\backslash S_{i-1}|} \tag{choice of greedy} \\ \vspace{0.5em}
%&\geq \frac{f(S^*)^\theta}{k} \cdot \frac{\sum_{e\in S^*\backslash S_{i-1}}f_{S_{i-1}}(e)}{\left(\sum_{e\in S^*\backslash S_{i-1}}f_{S_{i-1}}(e)\right)^\theta} \tag{$(\theta,c)$-sharp}\\
&= \frac{\sum_{e\in S^*\backslash S_{i-1}}f^\delta_{S_{i-1}}(e)}{|S^*\backslash S_{i-1}|}  \notag \\
&\geq \frac{f(S^*)^\theta}{kc}\cdot \left(\sum_{e\in S^*\backslash S_{i-1}}f^\delta_{S_{i-1}}(e)\right)^{1-\theta} \tag{sharpness} \\ \vspace{0.5em}
&\geq \frac{f(S^*)^\theta}{kc}\cdot [f(S_{i-1}\cup S^*)-f(S_{i-1}) \notag \\
&\quad + \delta |S^*\backslash S_{i-1}| f(S_{i-1})]^{1-\theta} \tag{submodularity}\\
&\geq \frac{f(S^*)^\theta}{kc}\cdot [f(S_{i-1}\cup S^*)-(1-\delta)f(S_{i-1})]^{1-\theta} \tag{$|S^*\backslash S_{i-1}| \geq 1$}\\
&\geq \frac{f(S^*)^\theta}{kc}\cdot [f(S^*) - (1-\delta)f(S_{i-1})]^{1-\theta}. \tag{monotonicity}
\end{align}
Therefore, we need to solve the following recurrence
\begin{equation}\label{eq:delta_rec} a_i \geq (1-\delta)a_{i-1} + \frac{a^\theta}{kc}\cdot [a - (1-\delta)a_{i-1}]^{1-\theta}\end{equation}
where $a_i =f(S_i)$, $a_0=0$ and $a=f(S^*)$.
%In Lemma  \ref{lemma:recurrence1} in Appendix \ref{sec:extra_proofs} we show that the solution of this recurrence satisfies for all $i\in\{1,\ldots,k\}$
%\[a_i\geq \frac{a}{1-\delta +\delta kc}\cdot \left[1-(1-\delta)^{\frac{1}{\theta}}\left(1-\frac{\theta }{c}\cdot\frac{i}{k}\right)^{\frac{1}{\theta}}\right],
%\]
%which finishes the proof since $f(S^g) = a_k$.

Define $h(x) = (1-\delta)x+ \frac{a^\theta}{kc}\cdot [a-(1-\delta)x]^{1-\theta}$, where $x\in[0,a]$. Observe that
\(h'(x)=(1-\delta)-\frac{a^\theta(1-\theta)(1-\delta)}{kc}\cdot[a-(1-\delta)x]^{-\theta}.\)
Therefore, $h$ is increasing inside the interval \[I:= \left\{x: \ 0\leq x\leq \frac{a}{1-\delta}\cdot \left(1-\left(\frac{1-\theta}{kc}\right)^{1/\theta}\right) \right\}.\] Let us define
\[b_i := \frac{a}{1-\delta +\delta kc}\cdot \left[1-(1-\delta)^{\frac{1}{\theta}}\left(1-\frac{\theta }{c}\cdot\frac{i}{k}\right)^{\frac{1}{\theta}}\right] .\]
First, let us check that $b_i\in I$ for all $i\in\{1,\ldots, k\}$. Namely, for any $i$ we need to show that
\begin{equation*}
\frac{a}{1-\delta +\delta kc}\cdot \left[1-(1-\delta)^{\frac{1}{\theta}}\left(1-\frac{\theta }{c}\cdot\frac{i}{k}\right)^{\frac{1}{\theta}}\right] \leq   \frac{a}{1-\delta}\cdot \left(1-\left(\frac{1-\theta}{kc}\right)^{1/\theta}\right).
\end{equation*}
Since $1-\delta \leq 1-\delta +\delta kc$, then it is sufficient to prove that
\begin{equation*}
1-(1-\delta)^{\frac{1}{\theta}}\left(1-\frac{\theta }{c}\cdot\frac{i}{k}\right)^{\frac{1}{\theta}} \leq   1-\left(\frac{1-\theta}{kc}\right)^{1/\theta} \quad \Leftrightarrow \quad (1-\delta)(kc-i\theta)\geq 1-\theta
\end{equation*}
The expression $kc -i \theta$ is decreasing on $i$. Hence, we just need the inequality for $i=k$, namely $(1-\delta)k(c-\theta)\geq 1-\theta$, which is true since $c\geq1$, $k\geq 1$ and $\delta\leq 1-\frac{1}{k}$.

Our goal is to prove that $a_i\geq b_i$, so by induction let us assume that $a_{i-1}\geq b_{i-1}$ is true. By using monotonicity of $h$ on the interval $I$, we get $h(a_{i-1})\geq h(b_{i-1})$. Also, observe that recurrence \eqref{eq:delta_rec} is equivalent to write $a_i \geq h(a_{i-1})$ which implies that $a_i \geq h(b_{i-1})$. To finish the proof we will show that $h(b_{i-1})\geq b_i$.

Assume for simplicity that $a=1$. For $x\in [1,k]$, define
 \[g(x) :=  (1-\delta)^{1/\theta} \left(1-\frac{\theta}{kc}\cdot x\right)^{1/\theta}.\]
 Note that $g'(x) = - \frac{1-\delta}{kc}\cdot g(x)^{1-\theta}$ and $g''(x) = \frac{(1-\delta)(1-\theta)}{(kc)^2}\cdot g(x)^{1-2\theta}$. Observe that $g$ is convex, so for any $x_1,x_2\in[0,k]$ we have $g(x_2)\geq g(x_1) + g'(x_1)\cdot (x_2-x_1)$. By considering $x_2 = i$ and $x_1 = i-1$, we obtain
\begin{equation}\label{eq:ineq_proof1}
g(i) - g(i-1) - g'(i-1) \geq 0
\end{equation}
On the other hand,
\begin{equation*} 
h(b_{i-1})- b_i=(1-\delta)[1-g(i-1)] + \frac{(1-\delta+\delta kc)}{kc}\left[1- \frac{1-\delta}{1-\delta+\delta kc}(1-g(i-1))\right]^{1-\theta} - 1 + g(i)
\end{equation*}
which is the same as
\begin{equation}
 h(b_{i-1})- b_i=g(i) - g(i-1) -\delta +\delta g(i-1) +  \frac{(1-\delta+\delta kc)^\theta}{kc}\left[\delta kc +(1-\delta)g(i-1)\right]^{1-\theta}. \label{eq:ineq_proof2}
\end{equation}

By H\"older's inequality we know that
\begin{equation*}
(1-\delta+\delta kc)^\theta\left[\delta kc +(1-\delta)g(i-1)\right]^{1-\theta}  \geq  \delta kc +(1-\delta)g(i-1)^{1-\theta}.
\end{equation*}
Then, by using this bound we get
\begin{align*}\eqref{eq:ineq_proof2} \quad &\geq \quad g(i) - g(i-1) -\delta +\delta g(i-1) +  \frac{1}{kc}\left[\delta kc +(1-\delta)g(i-1)^{1-\theta}\right] \\
& = \quad g(i) - g(i-1) - g'(i-1) + \delta g(i-1) \\
& \geq \quad 0,
\end{align*}
where we used the definition of $g'$ and inequality \eqref{eq:ineq_proof1}. This proves that $h(b_{i-1})- b_i \geq 0$, concluding the proof since $a_k=f(S^g)\geq b_k$.
\end{proof}

\bibliography{bibliography_updated}

\end{document}